\documentclass[sigconf]{acmart}
\pdfoutput=1

\let\ACMmaketitle=\maketitle
\renewcommand{\maketitle}{\begingroup\let\footnote=\thanks \ACMmaketitle\endgroup}

\renewcommand\footnotetextcopyrightpermission[1]{} 

\AtBeginDocument{%
  \providecommand\BibTeX{{%
    \normalfont B\kern-0.5em{\scshape i\kern-0.25em b}\kern-0.8em\TeX}}}

\usepackage{hyperref}
\newcommand{\secref}[1]{\S\ref{#1}}

\usepackage{tikz}
\usetikzlibrary{arrows, arrows.meta}
\usetikzlibrary{automata}
\usetikzlibrary{backgrounds}
\usetikzlibrary{shapes.misc, shapes.geometric}
\usetikzlibrary{shapes.gates.logic.US}
\usetikzlibrary{calc}
\usepackage{pgfplots}

\pgfdeclarelayer{edgelayer}
\pgfdeclarelayer{nodelayer}
\pgfsetlayers{background,edgelayer,nodelayer,main}
\usepackage{tikz-cd}

\usepackage{calc}
\usepackage {anyfontsize}
\usepackage{stmaryrd}
\SetSymbolFont{stmry}{bold}{U}{stmry}{m}{n}
\usepackage{color}
\usepackage{framed}
\usepackage{multirow}
\usepackage{subfigure}
\usepackage{arydshln}

\usepackage{listings}
\usepackage{amsmath}
\usepackage{amsthm}
\newtheorem{proposition}{Proposition}
\newtheorem{lemma}{Lemma}
\newtheorem{theorem}{Theorem}
\usepackage{blkarray}

\newcommand\set[1]{\{#1\}}
\newcommand\tuple[1]{\langle#1\rangle}

\newcommand\xor{\oplus}
\newcommand\ass{\mathrel{\gets}}
\newcommand\result[1]{\llbracket #1 \rrbracket}

\newcommand\gf[1]{\mathbb{F}_{#1}}
\newcommand\gfsize[2]{\mathbb{F}^{#2}_{#1}}

\newcommand\slp{\mathbb{SLP}}
\newcommand\xorslp{\slp_\xor}
\newcommand\multislp{\slp_{\vec{\xor}}}
\newcommand\return{\mathsf{ret}}
\newcommand\dcon[1]{\mathrel{\sim_{#1}}}
\newcommand\con[1]{\mathrel{\sim^{\!*}_{#1}}}
\DeclareSymbolFont{symbolsC}{U}{pxsyc}{m}{n}
\DeclareMathSymbol{\coloneqq}{\mathrel}{symbolsC}{"42}

\usepackage{stackengine}
\stackMath
\newlength\matfield
\newlength\tmplength
\def\matscale{1.}

\newcommand\raiserows[2]{%
   \setlength\matfield{\matscale\baselineskip}%
   \raisebox{#1\matfield}{#2}%
}

\newcommand{\lbxor}{{\textstyle\bigoplus}}

\newcommand\labelarrow[1]{
\xrightarrow{\smash{\raisebox{-1.5pt}{$\scriptstyle#1$}}}
}

\newcommand{\xRightarrow}[2][]{\underset{#1}{\overset{#2}{\Rightarrow}}}

\newcommand\costar[2][]{
    \xRightarrow[\smash{\raisebox{1pt}[0pt][0pt]{$\scriptstyle#1$}}]{\smash{\raisebox{-1pt}[0pt][0pt]{$\scriptstyle#2$}}}%
}
\newcommand{\xmultimap}[1]{\stackrel{\raisebox{-.5pt}[0pt][0pt]{$\scriptstyle#1$}}{\multimap}}
\newcommand{\xmultimapp}[2][]{\underset{\smash{\raisebox{1.5pt}[0pt][0pt]{$\scriptstyle#1$}}}{\overset{#2}{\multimap}}}

\DeclareMathSymbol\lbag\mathopen{stmry}{"2A}
\DeclareMathSymbol\rbag\mathclose{stmry}{"2B}

\newcommand{\plabel}[1]{{\ensuremath{\scriptstyle #1)}}}

\newcommand{\namedframe}[2]{
\begin{flushleft}
\tikzset{background rectangle/.style={thin,draw=black}}
\begin{tikzpicture}[every node/.style={inner sep=0,outer sep=0}, show background rectangle]
\node[align=justify, text width=0.45\textwidth, inner sep=3pt]{
#2
};
\node[inner sep=0em, xshift=0ex, yshift=-.1ex, overlay, above right, fill=white, draw=white]
at (current bounding box.north west) {
#1
};
\end{tikzpicture}
\end{flushleft}
}

\newcommand{\nvar}{\textsc{NVar}}
\newcommand{\cachecap}{\textsf{CCap}}
\newcommand{\iocost}{\textsf{IOcost}}
\newcommand{\val}[1]{\llparenthesis\,#1\,\rrparenthesis}

\newcommand{\bitmatrix}{\mathfrak{B}}

\newcommand\blfootnote[1]{%
	\begingroup
	\renewcommand\thefootnote{}\footnote{#1}%
	\addtocounter{footnote}{-1}%
	\endgroup
}

\acmConference[Technical Report]{Technical Report}{2021}{}
\startPage{1}
\setcopyright{none}

\begin{document}

\title{Accelerating XOR-based Erasure Coding using Program Optimization Techniques}

\author{Yuya Uezato}
\affiliation{
\institution{Dwango, Co., Ltd.}
\city{Tokyo}
\country{Japan}
}
\email{yuuya\_uezato@dwango.co.jp}


\newcommand{\appr}{\ensuremath{\mathord{\sim}}}

\begin{abstract}
Erasure coding (EC) affords data redundancy for large-scale systems.
XOR-based EC is an easy-to-implement method for optimizing EC.
This paper addresses a significant performance gap between the state-of-the-art XOR-based EC approach (\appr 4.9 GB/s coding throughput) and Intel’s high-performance EC library based on another approach (\appr 6.7 GB/s).
We propose a novel approach based on our observation that XOR-based EC virtually generates programs of a Domain Specific Language for XORing byte arrays.
We formalize such programs as straight-line programs (SLPs) of compiler construction
and optimize SLPs using various optimization techniques.
Our optimization flow is three-fold:
1) reducing operations using grammar compression algorithms;
2) reducing memory accesses using deforestation, a functional program optimization method; and
3) reducing cache misses using the (red-blue) pebble game of program analysis.
We provide an experimental library, which outperforms Intel’s library with a \appr 8.92 GB/s throughput. 
\end{abstract}

\maketitle

\blfootnote{
This is the author (and non-final) version of the same title paper that accepted by SC'21~\url{https://sc21.supercomputing.org/}.
The final version will be published by ACM in the DOI \url{https://doi.org/10.1145/3458817.3476204}.
The appendix contains a proof of Theorem~\ref{thm:intractability of minimum memory access problem} omitted from the conference version due to page limitation.
}

\section{Introduction}\label{sec:Introduction}
Assuring data redundancy is the most critical task for large-scale systems such as distributed storage.
Replication---distributing the replicas of data---is the simplest solution.
Erasure coding (EC) has attracted significant attention thanks to its space efficiency~\cite{Weatherspoon:2002}.
For example, the famous distributed system HDFS (Hadoop Distributed File System)~\cite{hdfs-ec} offers
the codec \textbf{RS}(10, 4), Reed-Solomon EC~\cite{Reed:1960} with 10 data blocks and 4 parity blocks.
On \textbf{RS}(10, 4), we can store 10-times more objects than through replication; however, we cannot recover data if five nodes are down.
Another distributed system, Ceph~\cite{ceph-ec}, offers $\textbf{RS}(n, p)$ for any $n$ and $p$.
On Linux, we can use RAID-6, a codec similar to $\textbf{RS}(n, 2)$~\cite{SNIA:spec:2009, Plank:2008}.
Using EC instead of the replication degrades the system performance since
the encoding and decoding of EC are heavy computation and are required for each storing to and loading from a system.
It is often stated that EC is suitable only for archiving cold (rarely accessed) data~\cite{Cook:2013, Huang:2012, Shenoy:2015}.

We clarify the pros and cons of EC by observing how $\textbf{RS}$ works.
To encode data using matrix multiplication (hereafter we use the acronym MM),
\textbf{RS} adopts matrices over $\gf{2^8}$, the finite field with $2^8 = 256$ elements.
Since each element of $\gf{2^8}$ is coded by one byte (8 bits), we can identify an $N$-bytes data as an $N$-array of $\gf{2^8}$.
$\textbf{RS}(n, p)$ encodes an $N$-byte data $D$ using
an $(n + p) \times n$ \emph{Vandermonde} matrix $\mathcal{V} \in \gfsize{2^8}{(n+p) \times n}$, which is key for decoding as we will see below,
as follows:
\[
\begin{tikzpicture}[baseline=(A)]
\draw[thick] (0, 0) rectangle (1, 1.7);
\draw[bend left=20] (0,0) to node[xshift=-5pt] {\small $\begin{array}{c}n\\[-3pt]+\\[-3pt]p\end{array}$} (0, 1.7);
\draw[bend right=20] (0, 1.7) to node[yshift=-5pt] {\small $n$} (1, 1.7);
\node (A) at (.5, .85) {$\mathcal{V}$};
\end{tikzpicture}
\cdot_{\gf{2^8}}
\raiserows{.8}{$\begin{pmatrix}
\vec{d}_1 \\[-2pt]
\scalebox{0.5}{$\vdots$} \\[-2pt]
\vec{d}_n
\end{pmatrix}$} =
\begin{pmatrix}
\vec{b}_1 \\[-2pt]
\scalebox{0.5}{$\vdots$} \\[-2pt]
\vec{b}_{n} \\[-2pt]
\scalebox{0.5}{$\vdots$} \\[-5pt]
\vec{b}_{n + p}
\end{pmatrix}
\ \ 
\begin{array}{l}
\text{where} \\
\quad\text{$\cdot_{\gf{2^8}}$ is the MM over $\gf{2^8}$;} \\[3pt]
\quad\text{$\vec{d}_i$ is $i$-th $\frac{N}{n}$-bytes block of $D$;} \\[3pt]
\quad\text{$\vec{b}_j$ is an $\frac{N}{n}$-bytes coded block.}
\end{array}
\]
We store an encoded block $\vec{b}_i$ to a node $n_i$ of a system with $n+p$ nodes.
For decoding, we gather $n$-blocks $B = (\vec{b}_{i_1} \vec{b}_{i_2} \ldots \vec{b}_{i_n})^T$ from alive nodes.
The $(n \times n)$-submatrix $\mathcal{M}$ of $\mathcal{V}$ obtained by extracting row-vectors at $\set{i_1, i_2, \ldots, i_n}$ satisfies $B = \mathcal{M} \cdot_{\gf{2^8}} D$. 
Since any square submatrix of Vandermonde matrices is invertible~\cite{MacWilliams:1977, Shilov:1977, Lang:1986},
the inverse $\mathcal{M}^{-1}$ of $\mathcal{M}$ recovers $D$ as $\smash{\mathcal{M}^{-1}} \cdot_{\gf{2^8}} B = D$.

Now, the advantage of \textbf{RS} \emph{space efficiency} emerges as the size of the encoded blocks.
For example, on $\textbf{RS}(10, 4)$, nodes of a system require $\frac{N}{10}$-bytes of space for an $N$-bytes data.
On the other hand, the disadvantage \emph{slowness} results from MM over finite fields.
Multiplying $n \times n$ matrices requires $\appr{}O(n^{2.37287})$ field operations even when using the latest result~\cite{Alman:2021, Gall:2014}.
Moreover, finite field multiplication $\times_{\gf{k}}$ is computationally expensive, and its optimization is an active research area~\cite{Koc:1998,Huang:2003,Kalcher:2011,Plank:2013,Larrieu:2019}.

There are two primary acceleration methods of \textbf{RS}.

\noindent \textbf{(1)} Tightly coupling sophisticated optimization methods for MM and finite field multiplication.
Intel provides an EC library, ISA-L (Intelligent Storage Acceleration Library), based on this approach~\cite{ISAL}.
ISA-L is exceptionally optimized for MM over $\gf{2^8}$
and offers different assembly codes for each platform to maximize SIMD instruction performance~\cite{IntelSIMD, AMDSIMD, ARMNeon}.
Intel reported ISA-L scored about \textbf{6.0 GB/s} encoding throughput for \textbf{RS}(10,4) in~\cite{ISALbenchmark}.
In our evaluation at \secref{sec:evaluation}, ISA-L scores for \textbf{6.7 GB/s}.

\noindent \textbf{(2)} XOR-based EC~\cite{Mastrovito:1989,blomer:1995,Zhou:2020}
converts MM over $\gf{2^8}$ to MM over $\gf{2}$ where $\gf{2}$ is the finite field of the bits $\set{0, 1}$.
This approach depends on the following two properties:\\
(i) There is an isomorphism $\bitmatrix: \gf{2^8} \cong \gfsize{2}{8 \times 1}$ between bytes and 8-bits column vectors; \\
(ii) There is a function $\widetilde{\cdot} : \gf{2^8} \to \gfsize{2}{8 \times 8}$ from bytes to $8 \times 8$ matrices over $\gf{2}$ such that:
$
\forall x, y \in \gf{2^8}.\ x \times_{\gf{2^8}} y = \bitmatrix^{-1}(\tilde{x} \cdot_{\gf{2}} \bitmatrix(y)).
$

We can calculate the above $\mathcal{V} \cdot_{\gf{2^8}} D$ without the finite field multiplication of $\gf{2^8}$, $\times_{\gf{2^8}}$, extending $\bitmatrix$ and $\tilde{\cdot}$ to matrices as follows:
\[
\mathcal{V} \cdot_{\gf{2^8}} D = \bitmatrix^{-1}(\tilde{\mathcal{V}} \cdot_{\gf{2}} \bitmatrix(D)).
\]

Since the addition (resp. multiplication) of $\gf{2}$ is the bit XOR $x \xor y$ (resp. bit AND),
MM over $\gf{2}$ is just array XORs, as presented below:
\[
\arraycolsep=0.3pt
\left(\begin{array}{lllllll}
1 & 1 & 0 & 0 & 0 & 0 & 0 \\
0 & 0 & 1 & 1 & 1 & 1 & 0 \\
0 & 0 & 1 & 1 & 1 & 0 & 1
\end{array}\right)
\cdot_{\gf{2}}
\begin{pmatrix}
\vec{d}_1 \\
\smash[t]{\vdots} \\
\vec{d}_7
\end{pmatrix} =
\begin{pmatrix}
\vec{d}_1 \xor \vec{d}_2 \\
\vec{d}_3 \xor \vec{d}_4 \xor \vec{d}_5 \xor \vec{d}_6 \\
\vec{d}_3 \xor \vec{d}_4 \xor \vec{d}_5 \xor \vec{d}_7
\end{pmatrix}.
\]

MM over $\gf{2}$ is easy-to-implement.
Thanks to its implementability, this method was proposed in VLSI to realize finite field arithmetic on a small circuit~\cite{Mastrovito:1989}.
This method is currently receiving attention since
array XORs $\vec{d}_i \xor \vec{d}_j$ are quickly executed via recent SIMD instructions~\cite{Plank:2013}.
In exchange for the ease of implementation, the obtained coding matrix $\tilde{\mathcal{V}} \in \gfsize{2}{8a \times 8b}$ is much larger than the original coding matrix $\mathcal{V} \in \gfsize{2^8}{a \times b}$;
thus, $\tilde{\mathcal{V}}$ needs more (but simple) operations of $\gf{2}$ than those of $\gf{2^8}$ in $\mathcal{V}$.

Recently, Zhou and Tian published an invaluable study~\cite{Zhou:2020} that synthesized several acceleration methods for executing array XORs.
It is the state-of-the-art study based on XOR-based EC; however, it scored \textbf{4.9 GB/s} for \textbf{RS}(10, 4) encoding.

Now, we have a question: ``Is XOR-based EC essentially slower than the former approach in exchange for the ease of implementation?''.
The answer is ``No''.
We provide a streamlined method to optimize XOR-based EC by employing various program optimization techniques.
We also implement and provide an experimental EC library outperforming ISA-L.

\section{Our Approach and Contribution}
We identify the MM over $\gf{2}$, $\cdot_{\gf{2}}$, as \emph{straight-line programs} (SLPs), a classical compiler theory tool~\cite{Aho:72, Aho:86}, as follows:
\[
\arraycolsep=0.3pt
\left(\begin{array}{lllllll}
1 & 1 & 0 & 0 & 0 & 0 & 0 \\
0 & 0 & 1 & 1 & 1 & 1 & 0 \\
0 & 0 & 1 & 1 & 1 & 0 & 1
\end{array}\right)
\cdot_{\gf{2}}
\begin{pmatrix}
\vec{a} \\
\vec{b} \\[1pt]
\smash[t]{\vdots} \\
\,\vec{g}\,
\end{pmatrix} \Mapsto
P:
\begin{array}{l}
\nu_1 \gets a \xor b; \\
\nu_2 \gets c \xor d \xor e \xor f; \\
\nu_3 \gets c \xor d \xor e \xor g; \\
\return(\nu_1, \nu_2, \nu_3)
\end{array}
\]
where $a, b, \ldots, g$ are constants meaning input arrays,
and $\nu_1, \nu_2, \nu_3$ are variables meaning arrays allocated at runtime.
SLPs are programs with a single binary operator without branchings, loops, and functions as above.

Replacing the MM over $\gf{2}$ by SLPs is a simple but key idea for importing various optimization methods from theory of programming.
This is the crucial difference between our study and that of Zhou and Tian~\cite{Zhou:2020},
where they treated topics directly on matrices of $\gf{2}$ and introduced ad-hoc constructions,
without sophisticated results of program optimization.

\paragraph{\textbf{Our Goal}}
The goal of this paper is to provide an efficient EC library importing various programmer-friendly optimization methods from theory of programming.
Technically, we implement our optimizer as a translator, which converts an SLP to a more efficient one.
When encoding and decoding data, we run optimized SLPs line-by-line in our host language in the interpreter style.

\subsection{Idea and Contribution in Our Optimizer}
We optimize SLPs via compressing, fusing, and scheduling.
Let us see the idea of each step by optimizing the above $P$ as follows:
\[
P
\!\underset{\textit{comp.}}{\Mapsto}\!
\begin{array}{r@{\ }l}
\lambda & \gets c \xor d \xor e; \\
\nu_1 & \gets a \xor b; \\
\nu_2 & \gets \lambda \xor f; \\
\nu_3 & \gets \lambda \xor g; \\
\multicolumn{2}{l}{\return(\nu_1, \nu_2, \nu_3)}
\end{array}
\!\!\!\!\!\!\!\!\underset{\textit{fuse}}{\Mapsto}\!
\begin{array}{r@{\ }l}
\lambda & \gets \lbxor (c, d, e); \\
\nu_1 & \gets a \xor b; \\
\nu_2 & \gets \lambda \xor f; \\
\nu_3 & \gets \lambda \xor g; \\
\multicolumn{2}{l}{\return(\nu_1, \nu_2, \nu_3)}
\end{array}
\!\!\!\!\!\!\!\!\!\!\!\!\underset{\textit{sched.}}{\Mapsto}\!
\begin{array}{r@{\ }l}
\nu_1 & \gets a \xor b; \\
\lambda & \gets \lbxor (c, d, e); \\
\nu_2 & \gets \lambda \xor f; \\
\lambda & \gets \lambda \xor g; \\
\multicolumn{2}{l}{\return(\nu_1, \nu_2, \lambda)}
\end{array}
\]

\paragraph{Compressing}
We use the compression algorithm \textsc{RePair}~\cite{Larsson:1999},
which is used to compress context-free grammars (CFGs) in grammar compression.
We can immediately adapt it by ignoring $\xor$ of SLPs, and by identifying constants (resp. variables) of SLPs as terminals (resp. nonterminals) of CFGs.

\textsc{RePair} compresses a program (or CFG) by extracting its hidden repetition structures.
For $P$, \textsc{RePair} extracts the repeatedly appearing subterm $c \xor d \xor e$ and replaces it with a new variable $\lambda$.
It reduces the seven XORs to five and speeds up $\frac{5}{7} \sim 30\%$.

We extend \textsc{RePair} to \textsc{XorRePair} by adding the XOR-cancellation property $(x \xor x \xor y = y)$, not considered in grammar compression.

\paragraph{Fusing}
To reduce memory access, we employ a technique called deforestation in functional program optimization~\cite{Wadler:1990}.
Deforestation eliminates intermediate data via fusing functions.
Although it has a deep background theory, we can easily adapt it thanks to the simplicity of SLP (one operator and no functions).

In our example, $c \xor d \xor e$ invokes six memory accesses
because $c \xor d$ invokes three---loading $c, d$ and storing the result to an \emph{intermediate} array $I_{c, d}$---,
and the remaining $I_{c, d} \xor e$ also invokes three.
By fusing the two XORs to $\bigoplus(c, d, e)$, we eliminate (deforest) the intermediate array $I_{c, d}$.
The fused XOR only invokes four memory accesses; loading $c$, $d$, and $e$, and storing the result array.

\paragraph{Scheduling}
To reduce cache misses, we revisit the well-known (but vague) maxim for cache optimization \emph{increasing the locality of data access}.
It appears in our example to reorder $\lambda$ and $\nu_1$ to adjust the generation site of $\lambda$ to the use sites, $\lambda \xor f$ and $\lambda \xor g$.
Furthermore, we reuse $\lambda$ without allocating and accessing $\nu_3$.

The maxim for cache optimization is too vague to automatically optimize SLPs and incorporate it into our optimizer.
Thus, in \secref{sec:pebble game}, 
we introduce measures for cache efficiency
and concretize our optimization problem as reducing the measures of a given SLP.
To optimize SLPs, we employ the (red-blue) pebble game of program analysis~\cite{Sethi:1975, Hong:81}.
The game is a simple abstract model of computation with fast and slow devices.
In our setting, the fast and slow devices are cache and main memory, respectively.

\paragraph{Performance}
Each step improves coding performance as follows.
\namedframe{\textit{Encoding Throughput Improvement on} \textbf{RS}(10, 4) (in \secref{sec:evaluation} \& \secref{evaluation:analyze throughput})}{
\begin{tikzpicture}[every node/.style={inner sep=0,outer sep=0}]
\node[anchor=south] (A) at (0.0, 0) {\begin{tabular}{c}Base:\\[1pt] 4.03GB/s\end{tabular}};
\node[anchor=south] (D) at (2.3, 0) {\begin{tabular}{c}In~\secref{sec:SLP}:\\[1pt] 4.36GB/s\end{tabular}};
\node[anchor=south] (E) at (4.5, 0) {\begin{tabular}{c}In~\secref{sec:Fusion}:\\[1pt] 7.50GB/s\end{tabular}};
\node[anchor=south] (F) at (6.8, 0) {\begin{tabular}{c}In~\secref{sec:pebble game}:\\[1pt] 8.92GB/s\end{tabular}};

\draw[->] (A) -- node[above]{\footnotesize Comp.}  (D);
\draw[->] (D) -- node[yshift=2pt, above]{\footnotesize Fuse} (E);
\draw[->] (E) -- node[yshift=2pt, above]{\footnotesize Sched.} (F);
\end{tikzpicture}
}
where Base runs unoptimized SLPs that are obtained from matrices,  such as the above $P$.
Interestingly, although \textsc{(Xor)RePair} reduces about 60\% XORs on average (as we will see in \secref{sec:evaluation}),
the summary says the compressing effect is small.
This is because \textsc{(Xor)RePair} generates cache-poor compressed SLPs,
and this observation will be substantiated by the cache analysis using our pebble game.
The sole application of \textsc{(Xor)RePair} is not good as the theoretical improvements suggest;
however, compressing achieves excellent performance in combination with memory and cache optimization.

\section{Related Work}
Zhou and Tian earnestly studied and evaluated various acceleration techniques for XOR-based EC \cite{Plank:2008, Huang:2007, Luo:2014} in \cite{Zhou:2020}. 
Their study comprises two stages: (i) reducing XORs of bitmatrices (= matrices over $\gf{2}$)~\cite{Huang:2007, Plank:2008};
and (ii) reordering XORs for cache optimization~\cite{Luo:2014}.
We emphasize that the previous works~\cite{Huang:2007, Plank:2008, Luo:2014}---and thus, Zhou and Tian---never considered SLPs, deforestation, and pebble games.
The lack of considering SLP makes a problem in each stage.
First, the XOR reducing heuristics of~\cite{Huang:2007, Plank:2008} run on graphs, which are obtained in an ad-hoc manner from bitmatrices.
This leads to a lack of considering the XOR-cancellation (unlike our \textsc{XorRePair}) and limited performance.
Indeed, Zhou and Tian reported the average reducing ratio (the smaller the better) $\frac{\text{\#XOR of reduced}}{\text{\#XOR of original}} \approx 65\%$.
Their ratio is larger than ours---$42.1\%$ of \textsc{RePair} and $40.8\%$ of \textsc{XorRePair}, as we will see in~\secref{sec:evaluation}.
Next, the cache optimization heuristics of~\cite{Luo:2014}, which reorders XORs locally without considering pebble games,
is not quite effective, $\frac{\text{throughput of optimized}}{\text{throughput of original}} \approx 101\%$ in~\cite{Zhou:2020}.
In contrast, our heuristics for the scheduling problem are effective, with an improvement ratio of $\approx 125\%$ in \secref{sec:evaluation}.

SLP has been widely studied in the early days of program optimization~\cite{Aho:72, Bruno:1976, Aho:77}.
Recently, Boyar~et~al.~revisited SLP with the XOR operator~\cite{Boyar:2008, Boyar:2013}
to optimize (compress) bitmatrices used in the field of cryptography, such as the AES S-box~\cite{Stoffelen:16, Reyhani:2018, Tan:2019}.
Their approach is based on Paar's heuristic~\cite{Paar:97}, which is almost the same as \textsc{RePair}~\cite{Larsson:1999}.
Previous works for cryptography~\cite{Boyar:2013, Kranz:2017, Reyhani:2018, Tan:2019}
focus on reducing the XORs in such special SLPs even if spending several days on one SLP.
Indeed, the proposed heuristics run in exponential time for aggressive optimization.
However, for \textbf{RS}(10, 4), as we will see in \secref{sec:evaluation},
we need to optimize 1002 SLPs for encoding and decoding.
On the basis of this difference, we propose the new heuristics \textsc{XorRePair} running in polynomial time.
Although the work of Boyar~et~al. inspired the authors,
we emphasize that they did not consider memory and cache optimization because their goal was to compress bitmatrices.

Hong and Kung proposed the red-blue pebble game~\cite{Hong:81} to model and study the transfer cost between fast and slow devices.
This game has been used to analyze a fixed algorithm and program, rather than for optimization.
There is a recent remarkable work by Kwasniewski~et~al.~where they used the pebble game to prove the near-optimality of their fixed MM algorithms~\cite{Kwasniewski:2019}.
Recently, there has been a trend to use the red-blue pebble game for program optimization~\cite{Carpenter:2016, Demaine:2018, Papp:2020}.
Our work is in this direction; indeed, the pebble game is the basis of our cache optimization algorithm.


\section{Reducing XOR Operations}\label{sec:SLP}
We formally introduce SLP with the XOR operator.
To optimize SLP, we employ a compression algorithm, \textsc{RePair}, and extend it to \textsc{XorRePair} by accommodating a property of XOR.
We will measure and compare the performance of \textsc{RePair} and \textsc{XorRePair} in~\secref{sec:evaluation}.

\subsection{Straight-Line Program}\label{sec:SLPdef}
A straight-line program (SLP) is a program without branchings, loops, and procedures~\cite{Aho:72, Bruno:1976, Aho:77}.
An SLP is a tuple $\tuple{\mathscr{V}, \mathscr{C}, \vec{s}, \vec{g}, \otimes}$
where
 $\mathscr{V}$ is a set of variables,
 $\mathscr{C}$ is a set of constants,
 $\vec{s}$ is a sequence of instructions (i.e., the body of the program),
 $\vec{g}$ is a sequence of variables returned by the program,
 and $\otimes$ is a binary operator.
The set of instructions $\tuple{\mathsf{instr}}$ is defined by the following BNF:
\[
\begin{array}{l}
\tuple{\mathsf{instr}} \coloneqq \mathscr{V} \gets \tuple{\mathsf{expr}} \\
\tuple{\mathsf{expr}} \coloneqq \mathscr{V} \mid \mathscr{C} \mid \tuple{\mathsf{expr}} \otimes \tuple{\mathsf{expr}}
\end{array}
\]

We consider a class of SLPs, \emph{XOR SLP}~\cite{Boyar:2008, Boyar:2013}, whose binary operator only satisfies the associativity ($(x \xor y) \xor z = x \xor (y \xor z)$),
commutativity ($x \xor y = y \xor x)$, and cancellativity ($x \xor x \xor y = y$) laws.
We write $\slp_\xor$ for the class.

$\slp_\xor$ is a DSL for XORing byte arrays.
For example, the following left $\slp_\xor$ abstracts the right array program:\\
{\centering
\begin{tabular}{l@{}l@{}l}
$\begin{array}{l}
v_1 \gets a \xor b; \\
v_2 \gets b \xor c \xor d; \\
v_3 \gets v_1 \xor v_2; \\
\return(v_2, v_3, v_1)
\end{array}
$&$\Leftrightarrow$&
\begin{lstlisting}[basicstyle=\small\ttfamily, columns=fullflexible]
P(a, b, c, d: [byte]) {
  var v$_1$ = a xor b;
  var v$_2$ = (b xor c) xor d;
  var v$_3$ = v$_1$ xor v$_2$;
  return (v$_2$, v$_3$, v$_1$);
}  
\end{lstlisting}
\end{tabular}}\\
where for the left SLP $\mathscr{V} = \set{ v_1, v_2, v_3 }$, $\mathscr{C} = \set{ a, b, c, d }$, and $\vec{g} = \tuple{v_2, v_3, v_1}$, and for the array program the infix function \verb|xor| performs XOR element-wise for input byte arrays.
On the basis of this idea, we have the following correspondence:
\[
\begin{array}{l}
\text{constants of $\slp_\xor$} \Leftrightarrow \text{program input arrays}, \\
\text{variables of $\slp_\xor$} \Leftrightarrow \text{arrays allocated at runtime}.
\end{array}
\]

\paragraph{Calculus on $\slp_\xor$}
We consider a set-based semantics where the value of a variable is a set of the constants.
We interpret $\xor$ as the symmetric difference of sets;
e.g., $\set{ a, b } \xor \set{ c, d } = \set{ a, b, c, d }$ and $\set{ a, b } \xor \set{ a, c } = \set{ b, c }$.
This semantics enables us to compute the above example $\slp_\xor$ as follows:
\[
\scalebox{0.95}{$
\begin{array}{l | c c c}
\hfill \text{SLP}\ P \hfill & v_1\text{-value} & v_2\text{-value} & v_3\text{-value} \\\hline
v_1 \gets a \xor b; & \set{a, b} & \\
v_2 \gets b \xor c \xor d; & \set{a, b} & \set{b, c, d} \\
v_3 \gets v_1 \xor v_2; & \set{a, b} & \set{b, c, d} & \set{a, c, d} \\
\return (v_2, v_3, v_1)
\end{array}
$}
\]

\textbf{Notation.}
We use $\result{\cdot}$ to denote the returned values of a program;
e.g., $\result{P} = \tuple{ \set{ b, c, d}, \set{a, c, d}, \set{a, b } }$.
We use $\#_\xor{\cdot}$ to denote the size of a program, i.e., the number of XORs;
e.g., $\#_\xor{P} = 4$.
We use $\mathsf{NVar}(\cdot)$ to denote the number of variables;
e.g., $\mathsf{NVar}(P) = |\set{v_1, v_2, v_3}| = 3$
where $|S|$ is the cardinality of a finite set $S$.

\subsection{Shortest SLP Problem}\label{sec:shortest SLP}
We formalize our first optimization problem.

\namedframe{\textbf{The shortest $\slp_\xor$ problem}}{
For a given $P \in \slp_\xor$,
we find $Q \in \slp_\xor$ that satisfies $\result{P} = \result{Q}$ and minimizes $\#_\xor Q$.
}

We cannot solve this problem in polynomial time unless \textbf{P=NP}
since the NP-completeness of its decision problem version was shown by Boyar~et~al~\cite{Boyar:2013}.
They reduced the NP-complete problem Vertex Cover Problem~\cite{Garey:1990} to the above problem.

\paragraph{Example: Minimizing via Cancellation.}
Let us consider the following three equivalent SLPs:
\[
\scalebox{0.95}{$
\begin{array}{lll}
\slp_\xor\ P_0 & \slp_\xor\ P_1 & \slp_\xor\ P_2 \\\hline
v_1 \gets a \xor b; & v_1 \ass a \xor b;  & v_1 \ass a \xor b; \\
v_2 \gets a \xor b \xor c; & v_2 \ass v_1 \xor c; & v_2 \ass v_1 \xor c; \\
v_3 \gets a \xor b \xor c \xor d; & v_3 \ass v_2 \xor d; & v_3 \ass v_2 \xor d; \\
v_4 \gets b \xor c \xor d; & v_4 \ass b \xor c \xor d ; & v_4 \ass v_3 \xor a;  \\
\return(v_1, v_2, v_3, v_4) & \return(v_1, v_2, v_3, v_4) & \return(v_1, v_2, v_3, v_4) \\
\end{array}
$}
\]
where $\result{P_0} = \result{P_1} = \result{P_2}$, $\#_\xor P_0 = 8$, $\#_\xor P_1 = 5$, and $\#_\xor P_2 = 4$.
We notice that, in $P_2$, the $\xor$-cancellativity is effectively used to compute $v_4$.
Indeed, we can show that there is no $Q$ with $\#_\xor Q < 4$ and $\result{P_0} = \result{Q}$ enumerating $Q \in \slp_\xor$, 
Moreover, there is no $Q$ with $\#_\xor Q< 5$ and $\result{P_0} = \result{Q}$ unless using the $\xor$-cancellativity.

This examples emphasizes the $\xor$-cancellativity is essential to shorten $\slp_\xor$.
We can also say that $P_2$ is 2x faster than $P_0$.

\subsection{Compressing SLP by \textsc{RePair}}\label{subsec:compressed by repair}
Instead of searching the shortest SLPs by tackling the intractable optimization problem,
we employ the grammar compression algorithm \textsc{RePair} from grammar compression theory as a heuristic.

In the original paper of \textsc{RePair}~\cite{Larsson:1999},
Larsson and Moffat applied a procedure called \emph{pairing} recursively to compress data (\textsc{RePair} stands for \emph{\textbf{re}cursive \textbf{pair}ing}).
For an SLP $P$ and a pair $(x, y)$ of terms (constants and variables),
we replace all the occurrences of the pair in $P$ introducing a fresh variable.
Hereafter, we call this step $\textsc{Pair}(x, y)$.
Let us apply $\textsc{Pair}(a, b)$ to the previous $\slp_\xor$ $P_0$.
\[
\scalebox{0.98}{$
\begin{array}{l}
v_1 \gets a \xor b; \\
v_2 \gets a \xor b \xor c; \\
v_3 \gets a \xor b \xor c \xor d; \\
v_4 \gets b \xor c \xor d; \\
\return(v_1, v_2, v_3, v_4)
\end{array}
\overset{\textsc{Pair}(a, b)}{\Mapsto}
\begin{array}{l}
t_1 \gets a \xor b; \\\hline
v_2 \gets t_1 \xor c; \\
v_3 \gets t_1 \xor c \xor d; \\
v_4 \gets b \xor c \xor d; \\
\return(t_1, v_2, v_3, v_4)
\end{array}
$}
\]
It replaces all $a \xor b$ with the new variable $t_1$ and reduces XORs.

To distinguish variables introduced by \textsc{Pair} and the others, we use horizontal lines as above.
Variables introduced by \textsc{Pair}, $t_1, t_2, \ldots$, are called \emph{temporals} and the others \emph{originals}; e.g., $t_1$ is temporal and $v_2, v_3, v_4$ are original.

To define our version of \textsc{RePair},
we need a total order $\prec$ on terms and extend it to the lexicographic ordering $\sqsubset$ on pairs.
In this paper, as an example, we use the total order defined as follows:
we order all constants in the alphabetical order and order temporal variables using their generation order: $t_1 \prec t_2 \prec \ldots \prec t_m$ where $t_i$ is generated before $t_{i+1}$ by \textsc{Pair}.
Furthermore, we require $t \prec c$ for a temporal variable $t$ and a constant $c$.

Now, we define \textsc{RePair} using \textsc{Pair} as its subroutine.
\namedframe{\textbf{\textsc{RePair}}}{
{\tikz[baseline=(A.base), remember picture] \node[inner sep=0pt, outer sep=0pt] (A) {\textbf{\textit{loop}}};} 
\begin{description}
\item[(1)] If there is no original variable, we terminate. 
\item[(2)] Otherwise, we choose a pair of terms that most frequently appears in the definitions of original variables (i.e., below the horizontal line).
We then apply \textsc{Pair} with the pair.
If there are multiple candidates, we select the smallest one for $\sqsubset$.
\end{description}
}
\begin{tikzpicture}[remember picture, overlay]
\draw[gray, ultra thick] (A) -- ($(A)+(0, -2.4)$);
\end{tikzpicture}
\textit{Example.}
Let us apply \textsc{RePair} to the above $P_0$ omitting $\return$:
\[
\hspace{-5pt}
\scalebox{0.98}{$
\begin{array}{l}
v_1 \gets a \xor b; \\ 
v_2 \gets a \xor b \xor c; \\
v_3 \gets a \xor b \xor c \xor d; \\
v_4 \gets b \xor c \xor d; \\
\end{array}
\!\!\overset{(a, b)}{\Mapsto}
\begin{array}{l}
t_1 \gets a \xor b; \\\hline
v_2 \gets t_1 \xor c; \\
v_3 \gets t_1 \xor c \xor d; \\
v_4 \gets b \xor c \xor d; \\
\end{array}
\overset{(t_1, c)}{\Mapsto}
\begin{array}{l}
t_1 \gets a \xor b; \\
t_2 \gets t_1 \xor c; \\\hline
v_3 \gets t_2 \xor d; \\
v_4 \gets b \xor c \xor d; \\
\end{array}
$}
\]
\[
\overset{(t_2, d)}{\Mapsto}
\scalebox{0.98}{$
\begin{array}{l}
t_1 \gets a \xor b; \\
t_2 \gets t_1 \xor c; \\
t_3 \gets t_2 \xor d; \\\hline
v_4 \gets b \xor c \xor d; \\
\end{array}
$}
\tikz[baseline=(X.base), remember picture] {\node[inner sep=0pt, outer sep=0pt] (X) {$\overset{(b, c)}{\Mapsto}$};}
\scalebox{0.98}{$
\begin{array}{l}
t_1 \gets a \xor b; \\
t_2 \gets t_1 \xor c; \\
t_3 \gets t_2 \xor d; \\
t_4 \gets b \xor c; \\\hline
v_2 \gets t_4 \xor d; \\
\end{array}
$}
\overset{(t_4, d)}{\Mapsto}
\scalebox{0.98}{$
\begin{array}{l}
t_1 \gets a \xor b; \\
t_2 \gets t_1 \xor c; \\
t_3 \gets t_2 \xor d; \\
t_4 \gets b \xor c; \\
t_5 \gets t_4 \xor d; \\\hline
\end{array}
$}
\]
At the first step, the pairs $(a, b)$ and $(b, c)$ appear three times, and we choose $(a, b)$ because $(a, b) \sqsubset (b, c)$.
The rest of the parts are processed in the same way.
We note that \textsc{RePair} reduces eight XORs to five, and the obtained SLP equals the previous $P_1$.

\subsection{New Heuristic: \textsc{XorRePair}}
We extend \textsc{RePair} by accommodating the XOR-cancellativity, which is not considered at all in \textsc{RePair}.

First, we introduce an auxiliary procedure, $\textsc{Rebuild}(v)$, which rewrites the definition of a given original variable $v$ using the values of temporal variables.
We also use the auxiliary notation $\val{w}$ to denote the value of a variable $w$.

\mbox{}
\namedframe{\textbf{\textsc{Rebuild}}($v$: original variable)}{
\textbf{Initialize}: Let $\mathit{rem} \coloneqq \val{v}$ and $\mathcal{S} \coloneqq \emptyset$. $\mathit{rem}$ denotes a set of constants to be eliminated by XORing existing temporal variables.

{\tikz[baseline=(A.base), remember picture] \node[inner sep=0pt, outer sep=0pt] (A) {\textbf{\textit{loop}}};} 
\begin{description}
\item[(1)] If we cannot shorten $\mathit{rem}$ (i.e., there is no temporal variable $t$ such that $| \mathit{rem}\,\xor\,\val{t} | < |\mathit{rem}|$),
we return $\mathit{rem} \cup \mathcal{S}$ as the new definition of $v$.
\item[(2)] Otherwise, we choose a temporal variable $t$ that minimizes $|\mathit{rem} \xor \val{t}|$
and update $\mathit{rem} \coloneqq \mathit{rem} \xor \val{t}$ and $\mathcal{S} \coloneqq \mathcal{S} \cup \set{ t }$.
If there are multiple candidates $t$, we choose the smallest one for $\prec$.
{\tikz[remember picture, overlay]\draw[gray, ultra thick] (A) -- ($(A)+(0, -2.85)$);}
\end{description}
}

For example, applying \textsc{Rebuild}($v_4$) to the following left SLP, we obtain a new equivalent definition $v_4 \gets a \xor t_3$:\\
\begin{tabular}{l|l}
\begin{minipage}{2.5cm}
\[
\begin{array}{l}
t_1 \gets a \xor b; \\
t_2 \gets t_1 \xor c; \\
t_3 \gets t_2 \xor d; \\\hline
v_4 \gets b \xor c \xor d; \\
\hfill \Downarrow \hfill \\
\hfill \set{a, t_3} \hfill
\end{array}
\]
\end{minipage}
&
\hspace{-15pt}
\begin{minipage}{8cm}
\begin{enumerate}
\item Set $\mathit{rem} = \set{ b, c, d } = \val{v_4}$.
\item Choose $t_3$ because
\begin{description}
\item{$t_1$:} $|\mathit{rem} \xor \val{t_1}| = |\set{a, c, d}| = 3$;
\item{$t_2$:} $|\mathit{rem} \xor \val{t_2}| = |\set{a, d}| = 2$;
\item{$t_3$:} $|\mathit{rem} \xor \val{t_3}| = |\set{a}| = 1$;
\end{description}
\item Set $\mathit{rem} = \set{a}$ and $\mathcal{S} = \set{ t_ 3 }$.
\item Return $\set{a} \cup \set{ t_3 }$ because \\
\quad $| \mathit{rem} \xor \val{t_i}| > |\set{a}| (= 1)\ \ \forall i \in \set{1, 2, 3}.$
\end{enumerate}
\end{minipage}
\end{tabular}

Augmenting \textsc{RePair} with \textsc{Rebuilt}, we obtain \textsc{XorRePair}:
\namedframe{\textbf{\textsc{XorRePair} = \textsc{RePair} + \textsc{Rebuild}}}{
{\tikz[baseline=(A.base), remember picture] \node[inner sep=0pt, outer sep=0pt] (A) {\textbf{\textit{loop}}};} 
\begin{description}
\item[(1) and (2)] are the same as \textsc{RePair}.
\item[(3)] For each original variable $v$, if $\textsc{Rebuild}(v)$ is strictly smaller than the current definition of $v$,
we update $v$.{\tikz[remember picture, overlay]\draw[gray, ultra thick] (A) -- ($(A)+(0, -1.4)$);}
\end{description}
}

Let us apply \textsc{XorRePair} to the example $P_0$.
\[
\cdots \Mapsto
\scalebox{0.95}{$
\begin{array}{l}
t_1 \gets a \xor b; \\
t_2 \gets t_1 \xor c; \\
t_3 \gets t_2 \xor d; \\\hline
v_4 \gets b \xor c \xor d;
\end{array}
\overset{\makebox[0pt]{\textsc{Rebuild}}}{\Mapsto}\quad
\begin{array}{l}
t_1 \gets a \xor b; \\
t_2 \gets t_1 \xor c; \\
t_3 \gets t_2 \xor d; \\\hline
v_4 \gets a \xor t_3
\end{array}
\overset{(a, t_3)}{\Mapsto}
\begin{array}{l}
t_1 \gets a \xor b; \\
t_2 \gets t_1 \xor c; \\
t_3 \gets t_2 \xor d; \\
t_4 \gets a \xor t_3; \\\hline
\end{array}
$}
\]
First, we reach the above left form.
Next, we update $v_4$ as $v_4 \gets a \xor t_3$ since $\textsc{Rebuilt}(v_4) = \set{a, t_3}$ shortens the definition of $v_4$.
Finally, we perform $\textsc{Pair}(a, t_3)$ and obtain the shortest $\slp_\xor$ with 4 XORs as we have seen in~\secref{sec:shortest SLP}.
Clearly, \textsc{XorRePair} runs in polynomial time.

\paragraph{Related approaches.}
We note that pairing is sometimes called \emph{factoring} in the context of common subexpression elimination (CSE) of compiler construction~\cite{Breuer:1969}.
Combining algebraic properties for simplifying expressions with CSE, like \textsc{XorRePair}, has been naturally considered in compiler construction~\cite{Aho:86, Muchnick:1998}; however, primal methods to choose terms to be factored are elaborated.
We adopt \textsc{RePair} to simply implement our compressor.
The effectiveness of \textsc{RePair} is already known in grammar compression~\cite{Charikar:2005} and in the context of compressing matrices over $\gf{2}$ for cryptography~\cite{Kranz:2017}.


\section{Reducing Memory Access}\label{sec:Fusion}
We reduce memory accesses of SLPs by employing \emph{deforestation}, an optimization method of functional program~\cite{Wadler:1990, Gill:1993, Coutts:2007}.

Deforestation has deep theoretical backgrounds~\cite{Burstall:1977,Takano:1995,Wadler:1989,Wadler:1990, Gill:1993, Coutts:2007};
however, we need a trivial idea for SLPs.
Let us consider the following SLP and the corresponding program:
\begin{center}
\begin{tabular}{r@{\quad}c@{\quad}l}
$
\begin{array}{l}
v_1 \gets a \xor b; \\
v_2 \gets v_1 \xor c; \\
v_3 \gets v_2 \xor d; \\
\return(v_3)
\end{array}
$
&
$\Leftrightarrow$
&
\begin{lstlisting}[basicstyle=\small\ttfamily, columns=fullflexible]
program(a, b, c, d) {
  var out = ((a xor b) xor c) xor d;
  return(out);
}
\end{lstlisting}
\end{tabular}
\end{center}

As the reader might notice, \verb|program| makes two intermediate byte arrays, which correspond to $v_1$ and $v_2$.
Since these intermediate arrays are immediately released, we would like to eliminate them.
To this end, we rewrite \verb|program| to the following one fusing XORs:
\begin{lstlisting}[basicstyle=\small\ttfamily, columns=fullflexible]
Xor$_4$(a, b, c, d) {
  var out = Array::new(a.len());
  for i in 0..out.len():
    out[i]$\ $=$\ $((a[i]$\,$^$\,$b[i])$\,$^$\,$c[i])$\,$^$\,$d[i]; // ^ = byte XOR
  return(out);
}
\end{lstlisting}
$\texttt{Xor}_4$ does not only generate intermediate arrays, but also reduces memory accesses.
For $N$-bytes arrays $A$ and $B$, $A\,\texttt{xor}\,B$ invokes $3N$ memory accesses (loading $A$ and $B$, and writing the XORed result); thus,
\verb|program| invokes $9N$ memory accesses for $N$-bytes arrays.
On the other hand, $\texttt{Xor}_4$ invokes $5N$ memory accesses.

Now, we augment SLPs with variadic XOR operators to formalize fused XORs, like $\texttt{Xor}_4$.
We note that no SLP faithfully reflects $\texttt{Xor}_4$ since the formalization of SLPs only admits binary operators.

\subsection{MultiSLP and Memory Accessing Problem}\label{subsec:XOR fusion}
We extend $\xorslp$ to $\multislp$ by accommodating variadic XORs $\lbxor(\vec{t})$.
$\multislp$ can represent the above $\texttt{Xor}_4$ as follows:
\[
\plabel{1}\ v \gets \lbxor(a, b, c, d); \qquad \plabel{2}\ \return(v);
\]
We also impose on $\multislp$ that there is no nested XORs; indeed, we can remove them as $v \gets (x_1 \xor x_2) \xor x_3 \Mapsto v \gets \lbxor(x_1, x_2, x_3)$.

To formalize our memory access optimization problem,
we define the number of memory accesses $\#_M(P)$ for $P \in \multislp$ as follows:
\[
\#_M(P) = \sum \set{ n + 1 : v \gets \lbxor(t_1, t_2, \ldots, t_n) \in P }.
\]

\namedframe{\textbf{The minimum memory access problem}}{
For an $P \in \slp_\xor$,
we find $Q \in \multislp$ that satisfies $\result{P} = \result{Q}$ and minimizes $\#_M(Q)$.
}

\begin{theorem}\label{thm:intractability of minimum memory access problem}
The minimum memory access problem cannot be solved in polynomial time unless \textbf{P=NP}.
\end{theorem}
This problem is also intractable as the same  as the shortest $\slp_\xor$
We prove the intractability in Appendix~\secref{sec:appendixA} by reducing the Vertex Cover Problem (VCP) to the above one.
Although using the VCP is the same as the construction given by Boyar~et~al.~\cite{Boyar:2013}
to prove the intractability of the shortest $\slp_\xor$ problem,
we need to deeply analyze $\slp_\xor$ and $\multislp$ because
a normalization from SLPs to SLPs of the binary form as follows, which is the key in the proof of~\cite{Boyar:2013}, does not work well for $\multislp$ and $\#_M$:
\[
\begin{array}{l}
v \gets a \xor b \xor c;
\end{array}
\Mapsto
\begin{array}{l}
v' \gets a \xor b; \\
v \gets v' \xor c; \\
\end{array}
\]
The above normalization to the binary form does not change $\#_\xor$; however, it increases $\#_M$ ($4 \to 6$).
This normalization brought a significantly useful syntactic property on $\slp_\xor$ in~\cite{Boyar:2013}.
Since we cannot count on such the property, in Appendix~\secref{sec:appendixA}, we give a more detailed and elaborated construction.

\subsection{XOR Fusion}\label{sec:xorfusion}
We propose a heuristic, XOR fusion, which reduces memory accesses of a given $\slp_\xor$ by transforming it to $\multislp$.
\namedframe{\textbf{XOR fusion}}{
Repeatedly applying the procedure that if there is a variable $v$ used just once in the program, we unfold $v$ as follows:
\[
\begin{array}{l@{}l}
v & \gets \bigoplus(t_1, t_2, \ldots, t_n); \\
v' & \gets \bigoplus(\ldots, v, \ldots); \\
\end{array}
\Mapsto
\begin{array}{l}
v' \gets \bigoplus(\ldots, t_1, t_2, \ldots, t_n, \ldots); \\
\end{array}
\]
}

The following is an example of the XOR fusion:
\[
\begin{array}{l}
v_1 \gets a \xor b; \\
v_2 \gets v_1 \xor c; \\
v_3 \gets v_2 \xor d; \\
\end{array}
\Mapsto
\begin{array}{l}
v_2 \gets \bigoplus(a, b, c); \\
v_3 \gets v_2 \xor d; \\
\end{array}
\Mapsto
v_3 \gets \lbxor(a, b, c, d);
\]

The fusion reduces memory accesses, and the following holds.
\begin{theorem}
Let $P$ be an $\multislp$, and $Q$ be an $\multislp$ obtained by applying the XOR fusion to $P$.
Then, $\#_M(Q) < \#_M(P)$ holds.
\end{theorem}

\paragraph{Why do not unfold variables used more than once?}
Let us consider the following three SLPs 
where $A$ is a source SLP, $B$ is obtained one by compressing $A$,
and $C$ is obtained by fusing $A$.
\begin{center}
\begin{tikzpicture}
\node (A) at (0, 0) {$A:
\begin{array}{l}
v_2 \gets a \xor b \xor c \xor d \xor e \xor f; \\
v_3 \gets a \xor b \xor c \xor d \xor e \xor g;
\end{array}
$};

\node (B) at (5, -0.7) {$B:
\begin{array}{l}
v_1 \gets \lbxor(a, b, c, d, e); \\
v_2 \gets v_1 \xor f; \\
v_3 \gets v_1 \xor g; \\
\end{array}
$};

\node (X) at (0, -1.4) {$C:
\begin{array}{l}
v_2 \gets \lbxor(a, b, c, d, e, f); \\
v_3 \gets \lbxor(a, b, c, d, e, g);
\end{array}$};

\draw[->] ($(A.east)+(-.2, 0)$) -> node[xshift=-2pt, auto]{\small compress} ($(B.west)$);
\draw[dashed, ->] ($(B.south west)+(0, .5)$) -> node[xshift=-2pt, auto]{\small disallow} (X.east);
\draw[->] (A) -> node[auto]{\small fuse} (X);
\end{tikzpicture}
\end{center}
where $\#_M(A) = 30$ since one XOR issues three accesses, $\#_M(B) = 12$, and $\#_M(C) = 14$.
Therefore, if we would allow to unfold $v_1$ in $B$, then fusing increases memory accesses.
In other words, the fusion without the restriction uncompresses a given SLP too much.

On the other hand, this situation suggests that uncompressed but fused SLP may run quickly in the real situation.
Since compressing introduces extra variables as above, it may bring terrible effects on cache.
In the next section~\secref{sec:pebble game}, we will consider cache optimization. Furthermore, we will compare the coding throughputs of directly fused SLPs and fully optimized (compressed, fused, and cache optimized) ones in~\secref{sec:evaluation}.


\newcommand{\alloc}[1]{\xmultimap{#1}}

\section{Reducing Cache Misses}\label{sec:pebble game}
We proposed the XOR fusion to reduce memory accesses in the previous section.
We now go one step further and reduce cache misses.
To this end,
we first review a classical cache optimization technique, blocking,
and then formalize our cache optimization problem on the basis of the (red-blue) pebble game~\cite{Sethi:1975, Hong:81}.
We see that our optimization problem cannot be solved in polynomial-time (unless \textbf{P = NP}) and provide polynomial-time heuristics.

\subsection{Blocking Technique for Cache Reusing}
\label{what is blocking technique}
Since the size of a CPU cache is small with compared to that of main memory,
we can only put a few arrays if a given data to be encoded is large.
For example, if the size of L1 cache is 32KB, which is a typical L1 cache size,
and a user encodes 1MB data on $\textbf{RS}(10, 4)$,
the input data is divided into $8 \cdot 10$ arrays of $\frac{1\text{MB}}{80} \approx 12$KB;
therefore, the cache can only hold two arrays at once, and thus the cache performance becomes poor.

To hold many arrays in cache at once,
we use the established cache optimization technique \emph{blocking}, which splits large arrays into small blocks introducing a loop.
Let us perform blocking for the following example, where we split arrays into arrays of $\mathcal{B}$ bytes.
\begin{center}
\begin{tabular}{c|c}
Original Program & Blocked One (Blocksize is $\mathcal{B}$) \\
\begin{lstlisting}[basicstyle=\small\ttfamily, columns=fullflexible, xleftmargin=0pt]
main(...) {
 
  v$_1$ = xor(A, B);
  v$_2$ = xor(C, D);
  v$_3$ = xor(v$_1$, E, F);
  v$_4$ = xor(v$_3$, G, A);
  v$_5$ = xor(v$_1$, v$_3$, v$_4$);
   
  return (v$_2$, v$_4$, v$_5$);
}
\end{lstlisting} &
\begin{lstlisting}[basicstyle=\small\ttfamily, columns=fullflexible]
main(...) {
  for $i$ in 0..(A.len() / $\mathcal{B}$) {
    v$^{[i]}_1$ = xor(A$^{[i]}$, B$^{[i]}$);
    v$^{[i]}_2$ = xor(C$^{[i]}$, D$^{[i]}$);
    v$^{[i]}_3$ = xor(v$^{[i]}_1$, E$^{[i]}$, F$^{[i]}$);
    v$^{[i]}_4$ = xor(v$^{[i]}_3$, G$^{[i]}$, A$^{[i]}$);
    v$^{[i]}_5$ = xor(v$^{[i]}_1$, v$^{[i]}_3$, v$^{[i]}_4$);
  }
  return (v$_2$, v$_4$, v$_5$);
}
\end{lstlisting}
\end{tabular}
\end{center}
where $X^{[i]}$ is the $i$-th $\mathcal{B}$-bytes block of an array $X$; i.e.,
$X^{[0]} = X[0..\mathcal{B}]$,
$X^{[1]} = X[\mathcal{B}..2\mathcal{B}]$, and so on.

\paragraph{\textbf{Measures of Cache Efficiency}}
We consider two measures of cache efficiency.
As the first measure, we consider the minimum cache capacity $\mathsf{CCap}$ to avoid \emph{cache reloading} while computing a given program.
It is called cache reloading to load a certain block that is once spilled from cache from memory to cache again.
If we can transform a given program $P$ to an equivalent one $Q$ with $\mathsf{CCap}(Q) < \mathsf{CCap}(P)$, then we can say $Q$ is more cache efficient.
As the second measure, we consider $\mathsf{IOcost}$ the total number of I/O transfers between cache and memory.
To formalize these measures, we augment SLPs with abstract cache memory.

\subsection{SLP Augmented with Abstract LRU Cache}
To formalize measures for cache efficiency,
hereafter we use SLPs to represent the inside of the loop introduced by the blocking technique by forgetting indices.
For example, the above program is rewritten as the following SLP $P_{eg}$:
\[
P_{eg}: \quad
\begin{array}{ll}
1)\,v_1 \gets A \xor B; & 2)\,v_2 \gets C \xor D; \\
3)\,v_3 \gets \lbxor(v_1, E, F);  & 4)\,v_4 \gets \lbxor(v_3, G, A); \\
5)\,v_5 \gets \lbxor(v_1, v_3, v_4); & 6)\,\return(v_2, v_4, v_5); \\
\end{array}
\]

We introduce notations to operate cache through SLPs.
\paragraph{\textbf{Computation with Cache}}
We simply consider a cache $\mathcal{C}$ as an ordered sequence of blocks, $\mathcal{C} = \beta_1, \beta_2, \ldots, \beta_n$, where each block $\beta_i$ is just a variable or constant.
The rightmost (resp. leftmost) block represents the most (resp. least) recently used element.

Let us consider to execute an XOR $v \gets \lbxor(t_1, t_2, \ldots, t_k)$.
We require $\set{t_1, \ldots, t_k, v} \subseteq \mathcal{C}$ and then
change $\mathcal{C}$ in the following steps:
\begin{enumerate}
\item For the arguments, in the order $i = 1, 2, \ldots, k$,
we load $t_i$ to $\mathcal{C}$ if $t_i \notin \mathcal{C}$ or update the position of $t_i$ if $t_i \in \mathcal{C}$;
\item We then allocate $v$ in $\mathcal{C}$ if $v \notin \mathcal{C}$ or update the position of $v$ if $v \in \mathcal{C}$.
\end{enumerate}
If the cache is full and there is no room for loading or allocating,
we evict the LRU (least recently used) element in the cache.
This eviction corresponds to spilling or writing-back a cached block to memory.
It corresponds to the typical cache replacement policy \emph{LRU replacement policy}~\cite{Hennessy:2017}.

\paragraph{Example}
Let us run our example SLP $P_{eg}$ with a 10-capacity cache.
For the first XOR, we load $A$ and then $B$ and finally allocate $v_1$.
These operations change $\mathcal{C}$ as follows:
\[
\mathit{empty}\ \smash{\costar{A}}\ A\ \smash{\costar{B}}\ A\ B\ \smash{\alloc{v_1}}\ A\ B\ v_1
\]
where we use $\smash{\costar{\bullet}}$ to denote a loading from memory and $\smash{\xmultimap{\bullet}}$ to an allocation in the cache.

For the second XOR, we load $C$ and $D$ and allocate $v_2$ as follows:
\[
A\ B\ v_1\ \smash{\costar{C}\,\costar{D}\,\alloc{v_2}}\  A\ B\ v_1\ C\ D\ v_2.
\]

For the third XOR, we update the position of $v_1$ in the cache, load $E$ and $F$, and allocate $v_3$:
\[
A\ B\ v_1\ C\ D\ v_2\ \smash{\xrightarrow{v_1}} \smash{\costar{E}} \smash{\costar{F}} \smash{\alloc{v_3}}\ 
A\ B\ C\ D\ v_2\ v_1\ E\ F\ v_3
\]
where we use $\smash{\xrightarrow{\mathbf{.}}}$ to denote a position update in the cache.

For the fourth XOR, fetching the arguments ($v_3$, $G$, and $A$) changes $\mathcal{C}$ as follows:
\[
A\ B\ C\ D\ v_2\ v_1\ E\ F\ v_3\ \xrightarrow{v_3} \costar{G} \xrightarrow{A}
B\ C\ D\ v_2\ v_1\ E\ F\ v_3\ G\ A.
\]
Since the fetch makes the cache $\mathcal{C}$ full, we evict the LRU element $B$ and then allocate $v_4$ as follows:
\[
B\ C\ D\ v_2\ v_1\ E\ F\ v_3\ G\ A\ \smash{\costar[B]{}} \smash{\alloc{v_4}}\ 
C\ D\ v_2\ v_1\ E\ F\ v_3\ G\ A\ v_4
\]
where we write $\smash{\costar[\bullet]{}}$ for evictions from the cache to memory.

Finally, we change $\mathcal{C}$ as follows in the fifth XOR:
\[
C\ D\ v_2\ v_1\ E\ F\ v_3\ G\ A\ v_4\ \smash{\xrightarrow{v_1}}
\smash{\xrightarrow{v_3}}\smash{\xrightarrow{v_4}}\smash{\costar[C]{}}\smash{\alloc{v_5}}\ 
D\ v_2\ E\ F\ G\ A\ v_1\ v_3\ v_4\ v_5.
\]

Now, we introduce two notations for the cache efficiency of SLPs.
\namedframe{$\mathsf{\mathbf{CCap}}(P: \text{SLP})$}{
$\mathsf{CCap}(P)$ denotes the minimum cache capacity where we can run $P$ without cache reloading.
}

We can confirm $\mathsf{CCap}(P_{eg}) = 10$.
Indeed, if we use the cache of capacity 9,
we need to replace $A$ for $G$ in the fourth XOR $v_4 \gets \lbxor(v_3, G, A)$, and this replacement leads to reloading $A$ as follows:
\[
A\ B\ C\ D\ v_2\ v_1\ E\ F\ v_3\ \smash{\costar[A]{G}} \smash{\costar[B]{A}}\ 
C\ D\ v_2\ v_1\ E\ F\ v_3\ G\ A
\]
where we write $\costar[x]{y}$ for the replacement that evicts $x$ and loads $y$.

We also consider the number of I/O transfers required by SLPs.
\namedframe{$\mathsf{\mathbf{IOcost}}(P : \text{SLP}, c : \text{cache capacity})$}{
$\mathsf{IOcost}(P, c)$ denotes the number of I/O transfers issued while running $P$ with a cache of $c$-capacity.
There are two kinds of I/O transfers;
transfers from cache to memory, $\smash{\costar{\bullet}}$, and
transfers from memory to cache, $\smash{\costar[\bullet]{}}$.
}
It is clear that $\mathsf{IOcost}(P_{eg}, 10) = 7 (\text{of $\costar{\bullet}$}) + 2 (\text{of $\costar[\bullet]{}$}) = 9$.
This measure is useful when cache capacity is determined by hardware.
For example, it is one of the standard parameter on recent CPUs that cache size is 32KB and cache block size is 64B.
The cache of such CPUs can hold 512 blocks maximally, and thus we optimize $\mathsf{IOcost}(P, 512)$.

Hereafter, as an example, we consider that the cache can holds eight blocks maximally.
We can easily check $\mathsf{IOcost}(P_{eg}, 8) = 13$ running $P_{eg}$ with the cache of capacity 8.

\subsection{Optimizing SLP via Register Allocation}
To reduce $\cachecap$ and $\iocost$, we try register allocation by identifying cache (resp. memory) of our setting as registers (resp. memory) of the usual register allocation setting.
Register allocation basically consists from three phases~\cite{Chaitin:1981, Chaitin:1982, Briggs:1994, George:1996, Appel:2001}:
(1) the register assignment phase where we rename variables of a given program so that it has smaller variables;
(2) the register spilling phase where we insert instructions to move the contents of registers to/from memory if variables are many than actual registers;
(3) the register coalescing phase where we merge variables that has the same meaning in the syntactic or semantic way.

Using the standard graph-coloring register assignment algorithm, we can obtain the following SLP from $P_{eg}$:
\[
P_{\text{reg}}:
\begin{array}{ll}
1)\ v_1 \gets A \xor B;  & 2)\ v_2 \gets C \xor D; \\
3)\ v_3 \gets \lbxor(v_1, E, F); & 4)\ v_4 \gets \lbxor(v_3, G, A); \\
5')\ v_1 \gets \lbxor(v_1, v_3, v_4); & 6)\ \return(v_2, v_4, v_1);
\end{array}
\]
Unlike $5)$ of $P_{eg}$, 
in $5')$, we store the result $\lbxor(v_1, v_3, v_4)$ to $v_1$ instead of $v_5$ of $P_{eg}$
since $v_1$ is no more needed after $5')$.

Although register assignment reduces variables, $\mathsf{NVar}(P_{\text{reg}}) = 4$, and I/O transfers, $\iocost(P_{\text{reg}}, 8) = 12$,
it does not reduce $\cachecap$ since $\cachecap(P_{eg}) = \cachecap(P_{\text{reg}}) = 10$.
We note that register spilling is useless since the LRU replacement disallows to select cached elements to be evicted.
Register coalescing also does not make any sense at least in the above example.
These tell that the cache optimization for SLPs by register allocation is quite limited.

Below we employ another approach, where we rearrange statements and arguments in SLPs.
It should be noted that program rearrangement or (re)scheduling is beyond register allocation.

Although register allocation on SLPs is not so powerful as above, it enjoys the following properties:
\begin{itemize}
\item The register assignment problem of SLPs can be solved in polynomial time.
This comes from the relatively new result that the register assignment of programs in the SSA (static single assignment) form is tractable~\cite{Pereira:2005, Hack:2006, Bouchez:2007}.
A program is in the SSA form if each variable is assigned exactly once~\cite{Alpern:1988, Rosen:1988, Cytron:1991}.
Since there is no branching in SLPs, we can easily convert SLPs to SSA SLPs; thus, the register assignment problem for SLPs is tractable.
Of course, the problem for general programs is intractable~\cite{Chaitin:1981}.
\item The register coalescing problem for SSA SLPs is also tractable; indeed, the variable coalescing operation does not increase the required number of registers.
The problem for general programs is intractable~\cite{Bouchez:2007c, Grund:2007}.
\item Register spilling is useful in the case where we can select elements to be evicted from the abstract cache.
Then, on SSA SLPs, if each variable is used at most once, the minimum cost register spilling problem for SSA SLPs can be solved in polynomial time~\cite{Bouchez:2007b}. Without the constraint, the problem becomes intractable~\cite{Farach:2000}.
\end{itemize}

\subsection{Optimizing SLP via Pebble Game}
We employ the classical tool of program analysis \emph{pebble game} to make a given SLP cache friendly.
On this setting, we do not only rename variables as well as register assignment does but also reorder the entire program.
We first introduce computation graphs, which are arenas of the pebble game, and then review the pebble game.

\paragraph{\textbf{Computation Graph}}
We use directed acyclic graphs (DAGs) to represent the value dependencies of SLPs:
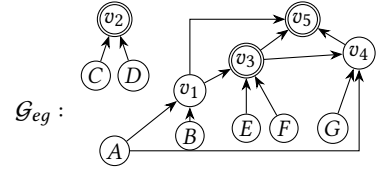
\begin{center}
\begin{tikzpicture}
\tikzstyle{goal}=[accepting, draw=black, shape=circle,  minimum size=3pt, inner sep=1pt]
\tikzstyle{new style 0}=[fill=white, draw=black, shape=circle, minimum size=10pt, inner sep=1pt]
\tikzstyle{new edge style 0}=[-Stealth]
\node at (-7, 3) {$\mathcal{G}_{eg}:$};	
	\begin{pgfonlayer}{nodelayer}
		\node [style=new style 0] (1) at (-5, 3.3) {$v_1$};
		\node [style=new style 0] (3) at (-6, 2.5) {$A$};
		\node [style=new style 0] (5) at (-5, 2.7) {$B$};
		\node [accepting, style=new style 0] (7) at (-6, 4.25) {$v_2$};
		\node [style=new style 0] (8) at (-6.25, 3.5) {$C$};
		\node [style=new style 0] (9) at (-5.75, 3.5) {$D$};
		\node [accepting, style=new style 0] (10) at (-4.25, 3.7) {$v_3$};
		\node [style=new style 0] (13) at (-4.25, 2.8) {$E$};
		\node [style=new style 0] (14) at (-3.75, 2.8) {$F$};
		\node [style=new style 0] (18) at (-2.75, 3.8) {$v_4$};
		\node [style=new style 0] (20) at (-3.10, 2.8) {$G$};
		\node [accepting, style=new style 0] (19) at (-3.5, 4.25) {$v_5$};
	\end{pgfonlayer}
	\begin{pgfonlayer}{edgelayer}
		\draw [style=new edge style 0] (3) to (1);
		\draw [style=new edge style 0] (5) to (1);
		\draw [style=new edge style 0] (8) to (7);
		\draw [style=new edge style 0] (9) to (7);
		\draw [style=new edge style 0] (13) to (10);
		\draw [style=new edge style 0] (14) to (10);
		\draw [style=new edge style 0] (1) to (10);
		\draw [style=new edge style 0] (10) to (18);
		\draw [style=new edge style 0] (3) -| (18);
		\draw [style=new edge style 0] (18) to (19);
		\draw [style=new edge style 0] (1) |- (19);
		\draw [style=new edge style 0] (10) to (19);
		\draw [style=new edge style 0] (20) to (18);
	\end{pgfonlayer}
\end{tikzpicture}
\end{center}
This DAG corresponds to our example $P_{eg}$ in the following sense.
Each leaf node (node with no children) represents the constant of the same name.
Each inner node (node with children) represents the value obtained by XORing all children;
thus, the inner node $v_1$ means $A \xor B$, $v_3$ means $A \xor B \xor E \xor F$, and so on.
\emph{Computation graphs (CGs)} are DAGs with double-circled nodes, \emph{goal nodes}, which mean values returned by programs.
It should be noted that CGs differ from \emph{interference graphs}, which are used in register allocation~\cite{Chaitin:1981, Chaitin:1982, Briggs:1994, George:1996, Appel:2001}.

\paragraph{\textbf{Pebble Game}}
Let $\mathcal{G}$ be a CG. As the initialization step, for each leaf node $\ell$, we put the same name pebble on $\ell$. To represent this, we write $\mathcal{G}(\ell) = \ell$.
We win a game if every goal node has a pebble.
To achieve this, we pebble inner nodes using the following rules:
\begin{itemize}
\item At each turn, the player proposes an instruction of the form
\[
n: p \gets \lbxor(p_1, p_2, \ldots, p_k)
\]
where $n$ is a node of $\mathcal{G}$, $p$ and $p_i$ are pebbles,
and $\set{ n_i : \mathcal{G}(n_i) = p_i }$ equals to $n$'s children.
\begin{itemize}
\item If $p$ is a new pebble, we put $p$ on $n$ so that $\mathcal{G}(n) = p$.
\item Otherwise, $p$ is in a node $m$, we \emph{move} $p$ from $m$ to $n$.
\end{itemize}
\item To avoid computing a single node multiple times, we are disallowed to put or move a pebble to a node once pebbled.
\end{itemize}

Our pebble game on CGs is equivalent to the standard pebble game of Sethi~\cite{Sethi:1973, Sethi:1975}.
Especially, it is equivalent to  the red-blue pebble game of Hong and Kung~\cite{Hong:81} that our game with the measure $\mathsf{IOcost}$ and the abstract cache $\mathcal{C}$ where we can select elements to be evicted from $\mathcal{C}$ instead of the LRU rule.

\paragraph{Example}
Let us consider the following winning strategy (with a return statement) of the above CG $\mathcal{G}_{eg}$:
\[
Q:
\begin{array}{ll}
1)\ v_1: p_1 \gets B \xor A; & 2)\ v_3: p_2 \gets \lbxor(E, F, p_1); \\
3)\ v_4: p_3 \gets \lbxor(A, G, p_2); & 4)\ v_5: p_1 \gets \lbxor(p_1, p_2, p_3); \\
5)\ v_2: p_3 \gets C \xor D; & 6)\ \return(p_3, p_2, p_1); \\
\end{array}
\]
This is better than $P_{\text{reg}}$ at all the parameters since
$\mathsf{NVar}(Q) = 3$, $\mathsf{CCap}(Q) = 5$, and $\mathsf{IOcost}(Q, 8) = 9$.
For example, we can easily confirmed $\mathsf{CCap}(Q) = 5$ as follows:
\[
\begin{array}{l}
\mathit{empty}\ \costar{B} \costar{A} \costar{p_1}\ B\ A\ p_1\ 
\costar{E} \costar{F} \xrightarrow{p_1} \xmultimapp[B]{p_2} A\ E\ F\ p_1\ p_2\ 
\xrightarrow{A} \costar[E]{G} \xrightarrow{p_2} \xmultimapp[F]{p_3} \\[6pt]
p_1\ A\ G\ p_2\ p_3\ \xrightarrow{p_1} \xrightarrow{p_2} \xrightarrow{p_3} \xrightarrow{p_1}\ 
A\ G\ p_2\ p_3\ p_1\ \costar[A]{C} \costar[G]{D} \xrightarrow{p_3}\ 
p_2\ p_1\ C\ D\ p_3
\end{array}
\]
where $\smash{\xmultimapp[x]{y}}$ means the replacement that evicts $x$ and allocates $y$.

The pebble game immediately implies the following property.

\begin{proposition}
Let $\mathcal{W}$ be a winning strategy of the CG of an SLP $P$.
Forgetting the node information from $\mathcal{W}$
and adding the adequate return statement,
we can obtain an SLP $Q_\mathcal{W}$ such that $\result{P} = \result{Q_\mathcal{W}}$.
\end{proposition}

\emph{\textbf{Notation}:}
Let $P$ and $Q$ be $\multislp$.
If $Q$ is obtained from a winning strategy of the CG of $P$ in the above manner,
we write $P \vdash Q$.

\subsection{Intractability of Optimization Problems}
Introducing the pebble game is not only useful for cache optimization
but also useful to correctly refer established results of compiler construction and program analysis.

\begin{theorem}\label{thm:intractability of problems of pebble game}
Let $P$ be an $\multislp$. All the following optimization problems cannot be solved in polynomial-time unless \textbf{P=NP}:
\begin{enumerate}
\item Finding $Q \in \multislp$ that satisfies $P \vdash Q$ and minimizes $\mathsf{NVar}(Q)$.
\item Finding $Q \in \multislp$ that satisfies $P \vdash Q$ and minimizes $\mathsf{CCap}(Q)$.
\item For a given cache capacity $c$,
finding $Q \in \multislp$ that satisfies $P \vdash Q$ and minimizes $\mathsf{IOcost}(Q, c)$.
\end{enumerate}
\end{theorem}

In order to show the intractability of Problems (1) and (2), we can use the NP-completeness of the standard pebble game shown by Sethi~\cite{Sethi:1973, Sethi:1975}.
Sethi reduced the classical NP-complete problem 3SAT~\cite{Cook:1971} to the decision problem of the standard pebble game.
Problem (1) is the optimizing version of the decision problem of the standard pebble game; hence, it is intractable.
Although Problem (2) seems a problem involved in cache, it is a problem of the standard pebble game rather than the pebble game with cache.
Indeed, if we could select pebbles to be evicted instead of the LRU rule,
Problems (1) and (2) are essentially equivalent. Even if we follow the LRU rule, the construction given by Sethi in~\cite{Sethi:1975} also works well; therefore, Problem (2) is intractable.

On the other hand, Problem (3) should be analyzed using the pebble game augmented with cache; namely, we use the red-blue pebble game of Hong and Kung~\cite{Hong:81}.
We can choose nodes to be evicted from the cache on the ordinal formalization of the red-blue pebble game. The intractability of Problem (3) on the red-blue pebble game was already shown in~\cite{ Demaine:2018, Papp:2020}.
In~\cite{Papp:2020}, Papp and Wattenhofer used the classic NP-complete problem Hamiltonian path problem~\cite{Karp:1972, Garey:1990} and succeeded in providing a simple NP-completeness proof.
Fortunately, we can directly apply the construction of Papp and Wattenhofer to Problem (3) in our setting with the LRU eviction rule. 

It is worth noting that Problem (1) and its variant can be efficiently solved when playing the pebble game on trees rather than DAGs~\cite{Ershov:1958, Nakata:1967, Sethi:1970, Schneider:1971, Lengauer:1980, Liu:1986}.

\subsection{Two Scheduling Heuristics}\label{sec:pebble heuristics}
We consider two simple heuristics for solving the pebble game
since our interested problems are intractable.
More technically, it is known that those problems are hard to approximate~\cite{Demaine:2017, Papp:2020}.

\paragraph{DFS-based algorithm.}
Our first heuristic visits the nodes of a given CG in the postorder traversal.

Let us see how our heuristic works for our CG $\mathcal{G}_{eg}$.
We need to decide which root node is visited first;
here, we choose $v_2$ on the basis of the total ordering $\prec$ defined in~\secref{subsec:compressed by repair} since $v_2 \prec v_5$.
We then visit the children $C$ and $D$ in this order since $C \prec D$.
Using $\prec$ as the tie-breaker, we make the following postorder traversing:
\[
C \to D \to v_2 \to A \to B \to v_1 \to E \to F \to v_3 \to G \to v_4 \to v_5.
\]

On the basis of this order, we generate a winning strategy as follows:
\[
Q_{\text{DFS}}:
\begin{array}{ll}
1)\ v_2: p_1 \gets C \xor D; & 2)\ v_1: p_2 \gets A \xor B; \\
3)\ v_3: p_3 \gets \lbxor(p_2, E, F); & 4)\ v_4: p_4 \gets \lbxor(p_3, A, G); \\
5)\ v_5: p_4 \gets \lbxor(p_2, p_3, p_4); & 6)\ \return(p_1, p_3, p_4);
\end{array}
\]
Our pebble assigning policy is simple.
If we have a pebble on $G$ that can move, we reuse it; otherwise, we put a fresh pebble.
It can be verified that $\mathsf{NVar}(Q_{\text{DFS}}) = 4$, $\cachecap(Q_{\text{DFS}}) = 7$,
and $\iocost(Q_{\text{DFS}}, 8) = 10$.

\paragraph{Bottom-up greedy algorithm.}
Our next heuristic is a greedy one.
Unlike the above DFS-based algorithm, this heuristic requires a parameter $c$ corresponding to cache capacity.
\begin{enumerate}
\item[(i)] Choose a computable node $n$, whose children have pebbles, that maximises the ratio $\frac{|H|}{|C|}$ where
$C$ are the children of $n$ and $H \subseteq C$ are the children whose pebble in the cache.
\item[(ii)] Access $H$ and then access $C$.
\item[(iii)] If there is a movable cached pebble, we move it to $n$.
Otherwise, we use a movable pebble or allocate a fresh pebble.
\end{enumerate}
Here we again use $\prec$ as the tie-breaker.

We revisit the CG $\mathcal{G}_{eg}$ as follows.
On the initial state, $v_1$ and $v_2$ are ready
with $\frac{|\emptyset|}{|\set{A, B}|} = \frac{0}{2}$ for $v_1$ and
$\frac{|\emptyset|}{|\set{C, D}|} = \frac{0}{3}$ for $v_2$.
We choose $v_1$ since $v_1 \prec v_2$, and the generated statement changes $\mathcal{C}$ as the following right:
\[
v_1: p_1 \gets A \xor B; \qquad \mathit{empty} \costar{A} \costar{B} \xmultimapp{p_1} A\ B\ p_1.
\]
Next, we choose $v_3$ since $v_3: \frac{|\set{v_1}|}{|\set{v_1, E, F}|} = \frac{1}{3}$ and 
$v_2:  \frac{|\emptyset|}{|\set{C, D}|} = \frac{0}{2}$, and compute $v_3$ with a fresh pebble $p_2$.
Repeating this procedure, we obtain the following sequences and an SLP $Q_{\text{greedy}}$:
\[
\begin{array}{ll}
v_3: p_2 \gets \lbxor(p_1, E, F); & \cdots \labelarrow{p_1} \costar{E} \costar{F} \xmultimap{p_2} A B p_1 E F p_2, \\
v_4: p_3 \gets \lbxor(p_2, A, G); & \cdots \xrightarrow{p_2} \xrightarrow{A} \costar{G} \alloc{p_3} B p_1 E F p_2 A G p_3, \\
v_5: p_1 \gets \lbxor(p_1, p_2, p_3); & \cdots \xrightarrow{p_1} \xrightarrow{p_2} \xrightarrow{p_3} \xrightarrow{p_1} B E F A G p_2 p_3 p_1, \\
v_2: p_3 \gets C \xor D; & \cdots \costar[B]{C} \costar[E]{D} \xrightarrow{p_3} E F A G p_2 p_1 C D p_3.
\end{array}
\]
It can be verified that $\mathsf{NVar}(Q_{\text{greedy}}) = 3$, $\cachecap(Q_{\text{greedy}}) = 7$,
and $\iocost(Q_{\text{greedy}}, 8) = 9$.
The scores of $\mathsf{NVar}$ and $\iocost$ are optimal.

\setlength{\belowcaptionskip}{-10pt}

\section{Evaluation and Discussion}\label{sec:evaluation}
We evaluate our optimizing methods.
In~\secref{evaluation:dataset}, we explain our dataset.
In~\secref{evaluation:unotpmize on various block sizes}, we see throughputs of an unoptimized SLP on different block sizes.
In~\secref{evaluation:average}, we show the average performance of \textsc{(XOR)Repair} of~\secref{sec:SLP}, the XOR fusion of~\secref{sec:xorfusion}, and scheduling heuristics of~\secref{sec:pebble heuristics}.
In~\secref{experiment:effect of blocksize}, we tell how the block size of the blocking technique affects coding performance.
In~\secref{evaluation:analyze throughput}, we show coding throughputs of programs fully optimized by our methods.
In~\secref{evaluation:throughput comparison}, we compare our throughputs with Intel's ISA-L~\cite{ISAL} and the state-of-the-art study~\cite{Zhou:2020}.

All experiments are conducted on the following environments:
\begin{center}
\begin{tabular}{c|ccccccc}
name & CPU & Clock & Core &  RAM \\\hline
\textbf{intel} & i7-7567U & 4.0GHz & 2 & DDR3-2133 16GB \\
\textbf{amd} & Ryzen 2600 & 3.9GHz & 6 & DDR4-2666 48GB \\
\end{tabular}
\end{center}
The cache specification of these CPUs are the same; the L1 cache size is 32KB/core, the L1 cache associativity is 8, and the cache line size is 64 bytes. 
Our EC library and codes to reproduce the results in this section can be found in~\cite{our_repository}.
Our library is written by Rust and compiled by rustc-1.50.0. 

\textbf{Important Remark:} We select the above environments for the following reason.
Since it has not been opened that the source codes implemented and used in the study of Zhou and Tian~\cite{Zhou:2020},
we cannot directly compare our methods and theirs by running programs.
Thus, we borrow values from~\cite{Zhou:2020} and compare them with our results measured on the above environments, which close to theirs
Intel i7-4790(4.0 GHz, 4 cores, 32KB cache, 64-bytes cache block, 8-way assoc.) and
AMD Ryzen 1700X(3.8 GHz, 8 cores, 32KB cache, 64-bytes cache block, 8-way assoc.).

\subsection{Dataset}\label{evaluation:dataset}
As an evaluation dataset, we use matrices of the codec $\textbf{RS}(10, 4)$, which is used in Hadoop HDFS~\cite{hdfs-ec} as stated in~\secref{sec:Introduction}.
We have 1002 coding matrices---one encoding matrix and ${14 \choose 4} = 1001$ decoding matrices obtained by removing 4 rows from the encoding matrix.
We need the finite field $\gf{2^8}$ to make a Vandermonde matrix for coding, as we have seen in~\secref{sec:Introduction}.
We implemented it in our experimental library on the basis of the standard construction.

Technically speaking (to readers who are familiar with coding theory),
we implement $\gf{2^8}$ using the primitive polynomial $x^8+x^4+x^3+x^2+1$ used in ISA-L~\cite{ISAL}.
For $\textbf{RS}(10, 4)$, to use the same encoding matrix of ISA-L,
we adopt the reduced form (aka, standard form~\cite{MacWilliams:1977,Ling:2004,Lacan:2004}) $\mathcal{V}$ of a $(14 \times 10)$ Vandermonde matrix given by the standard construction as follows:
\[
\arraycolsep=2pt
\left(
\begin{array}{lclc}
1 & \alpha & \cdots & \alpha^{9} \\
1 & \alpha^2 & \cdots & (\alpha^2)^{9} \\
\vdots & \vdots & & \vdots \\
1 & \alpha^{14} & \cdots & (\alpha^{14})^{9}
\end{array}
\right)
=
\left(
\begin{array}{c}
V_{10 \times 10} \\\hline
M_{4 \times 10}
\end{array}
\right)
\overset{\text{reduce}}{\Mapsto} \mathcal{V} = 
\left(
\begin{array}{c}
\text{Ident}_{10 \times 10} \\\hline
M_{4 \times 10} V^{-1}_{10 \times 10}
\end{array}
\right)
\]
where $\alpha$ is a primitive element of our $\gf{2^8}$~\cite{Reed:1960}, and the reduced version $\mathcal{V}$ is the actual encoding matrix of ISA-L.

We write $P_{\text{enc}}$ for the SLP that corresponds to the bitmatrix form $\tilde{\mathcal{V}}$ of
our encoding matrix $\mathcal{V}$, as seen in~\secref{sec:Introduction}.
We write $\mathcal{P}_{\text{RS}}$ for the sets of all the SLPs corresponding to the coding matrices.

\subsection{Performance of Unoptimized \texorpdfstring{$P_{\text{enc}}$}{Penc} on Various Block Sizes}
\label{evaluation:unotpmize on various block sizes}
As we have seen in~\secref{sec:SLPdef}, the execution of SLP is executing array XORs.
Since we apply the blocking technique of~\secref{what is blocking technique} to exploit cache,
we prepare two procedures for XORing blocked arrays of size $\mathcal{B}$.
The first one \verb|xor1| performs the byte XORing element-wise.
The second one \verb|xor32| performs the 32byte XORing---\verb|mm256_xor|, SIMD AVX2 instruction---element-wise.
Such SIMD instructions are used in ISA-L and the previous study~\cite{Zhou:2020}.
Furthermore, AVX2 is the standard instruction set for recent CPUs.

\begin{tabular}{@{\hspace{-10pt}}ll}
\begin{lstlisting}[basicstyle=\small\ttfamily, columns=fullflexible]
$\textsf{xor1}$(a$_1$, a$_2$) {
 var out = Array::new($\mathcal{B}$);
 for(i=0; i<$\mathcal{B}$; i+=1):
   out[i] = a$_1$[i] xor a$_2$[i];
 return out;
}
\end{lstlisting}&
\begin{lstlisting}[basicstyle=\small\ttfamily, columns=fullflexible]
$\textsf{xor32}$(a$_1$, a$_2$) {
 var out = Array::new($\mathcal{B}$);
 for(i=0; i<$\mathcal{B}$; i+=32):
   (out[i] as m256)
      = mm256_xor(&a$_1$[i], &a$_2$[i]);
 return out;
}
\end{lstlisting}
\end{tabular}\\
where \texttt{m256} is the type of 32 bytes in AVX2.
In the same way, we implement $n (>2)$ arity versions for running fused SLPs.

To measure the performance of our unoptimized SLP $P_{\text{enc}}$, we execute it for randomly generated arrays of 10MB.
The following table is the average throughput (GB/sec) of 1000-times executions.

\begin{center}
\begin{tabular}{@{\hspace{0pt}}c|c@{\hspace{8pt}}|c@{\hspace{8pt}}c@{\hspace{8pt}}c@{\hspace{8pt}}c@{\hspace{8pt}}c@{\hspace{8pt}}c@{\hspace{8pt}}c}
\begin{tabular}{c}
{\small Throughput}\\
{\small (GB/sec)}
\end{tabular} & \textsf{xor1} & \multicolumn{7}{|c}{\textsf{xor32}} \\\hline
{\small Blocksize} & & 64 & 128 & 256 & 512 & 1K & 2K & 4K \\\hline\hline
\textbf{intel} & 0.16 & 0.62 & 1.12 & 2.05 & 3.02 & 4.03 & 4.78 & 4.72  \\\hline
\textbf{amd} & 0.17 & 0.67 & 1.17 & 1.72 & 2.16 & 2.78 & 3.17 & 3.29  \\\hline
\end{tabular}
\end{center}
The power of SIMD is remarkable, as already reported in~\cite{Plank:2013, Zhou:2020}.
This result also suggests that the bottleneck is shifted from the CPU to memory I/O by the SIMD instruction.

Despite our argument about cache efficiency in~\secref{sec:pebble game},
the performances of small blocks—64, 128, 256, 512— are worse than those of large blocks, 1K, 2K, and 4K.
It is well-known in the context of cache optimization that a too-small block is not good in real computing~\cite{Lam:1991, Coleman:1995, Rivera:1999, Yotov:2003}.
We will evaluate and discuss how the change of block sizes affects coding performance below in~\secref{experiment:effect of blocksize}.

\subsection{Average Reduction Ratios of Our Methods}\label{evaluation:compression}\label{evaluation:average}
\paragraph{Reducing Operators}
We evaluate our SLP compression heuristics, \textsc{RePair} and \textsc{XorRePair}.
The following table displays the average performance of \textsc{(Xor)RePair} for the 1002 SLPs of \textbf{RS}(10, 4):
\[
\begin{array}{c|cc:c}
\text{Avg}\% & \frac{\textsc{Repair}(P)}{P} & \frac{\textsc{XorRepair}(P)}{P} &
\begin{tabular}{c}
\text{\small Corresp. Value}\\
\text{\small from~\cite{Zhou:2020}}
\end{tabular}
\\\hline
\text{\footnotesize XOR Num.}~~\#_\xor(\cdot) & 42.1\% & 40.8\% & \appr65.0\%
\end{array}
\]
where the first and second ratios are the average ratios of reducing XORs by our heuristics defined as follows;
\[
\textit{Avg}\left(\left\{ \frac{\#_\xor \mathcal{C}(P)}{\#_\xor P} : P \in \mathcal{P}_{\text{RS}} \right\}\right)
=
\begin{cases}
42.1\% & \text{ if $\mathcal{C} = \textsc{Repair}$}, \\
40.8\% & \text{ if $\mathcal{C} = \textsc{XorRepair}$}.
\end{cases}
\]
We note that the smaller the ratio, the better the compressing performance.
The value 65.0\% is the best ratio among the XOR reduction heuristics for bitmatrices evaluated in~\cite{Zhou:2020}.
Although \textsc{RePair} is simple and developed initially in grammar compression, we can see it works very well.

This table also says \textsc{XorRePair} exploits the cancellative property of XOR; but, the difference is minor.
It is not surprising; indeed, exponential-time compression heuristics and an algorithm, which corresponds to \textsc{RePair}, were already compared in~\cite{Kranz:2017} for the application to cryptography, and there was also little difference.
These results mean that \textsc{RePair} efficiently compresses programs, even though it does not use the cancellativity of XOR.
We consider this comes from the robustness of RePair,
which also appears in grammar compression when comparing it with other compression algorithms, such as LZ77 and LZ78~\cite{Charikar:2005}.

\paragraph{Reducing Memory Access}
We see how \textsc{XorRepair} and the XOR fusion of~\secref{subsec:XOR fusion} reduce memory access $\#_M(\cdot)$:
\[
\begin{array}{c|cccc}
\text{Avg\%} & \frac{\textsc{Co}(P)}{P} & \frac{\textsc{Fu}(P)}{P} & \frac{\textsc{Fu}(\textsc{Co}(P))}{\textsc{Co}(P)} & \frac{\textsc{Fu}(\textsc{Co}(P))}{P} \\\hline
\#_M(\cdot) & 40.8\% & 35.1\% & 59.2\% & 24.1\% \\
\end{array}
\]
where \textsc{Co} means \textsc{XorRepair}, and \textsc{Fu} means the XOR fusion.

The second ratio says that the XOR fusion averagely reduces ${\sim}65\%$ memory accesses for \emph{uncompressed} SLPs.
We can see the other columns in the same way.
Therefore, we can tell that \textsc{XorRepair} and the XOR fusion work well independently; furthermore, combining them averagely reduces ${\sim}76\%$ memory accesses on average.

\paragraph{Reducing Variables and Required Cache Size}
We consider how the XOR fusion and our DFS-scheduling heuristic
averagely affect the two measures of the cache efficiency $\nvar$ and $\cachecap$.
\[
\begin{array}{c|cccc}
\text{Avg}\% & \frac{\textsc{Co}(P)}{P} & \frac{\textsc{Fu}(P)}{P} &
\frac{\textsc{Fu}(\textsc{Co}(P))}{\textsc{Co}(P)} &\frac{\textsc{Dfs}(\textsc{Fu}(\textsc{Co}(P)))}{\textsc{Co}(P)} \\\hline
\nvar & 1552 \% & 100 \% & 38.9\% & 24.5 \% \\\hdashline
\cachecap & 498\% & 98.7 \% & 51.2\% & 40.0 \% \\
\end{array}
\]
where $\textsc{Dfs}$ means our DFS-based scheduling heuristic.
We skip using our greedy-scheduling heuristic and the measure $\iocost(\_, \_)$
since they depend on the cache sizes determined by our block size $64, 128, \ldots, 4K$, and the table including values for of all the cache sizes becomes too large. 

The first ratio clarifies that \textsc{XorRepair} significantly degrade cache efficiency.
Comparing the third and fourth ratios, we can say the scheduling heuristic certainly improves cache efficiency.
Multiplying the first and fourth ratios derives
$\frac{\cachecap(\textsc{Dfs}(\textsc{Fu}(\textsc{Co}(P))))}{\cachecap(P)} \sim 199\%$; therefore, we can say that the scheduling heuristic can suppress the side effects of \textsc{XorRepair} to some extent.

We consider why \textsc{XorRepair} significantly deteriorates $\nvar$ and $\cachecap$.
It results from the intrinsic behaviors of \textsc{(Xor)RePair}; namely, they add many temporal variables without considering cache and register efficiency.
The same inefficiency problem was pointed in the early research of program optimization as the weak point of CSE~\cite{Allen:1972}.

\subsection{Selecting Adequate Blocksize}\label{experiment:effect of blocksize}
As we have seen in~\secref{sec:pebble game}, the block size $\mathcal{B}$ of the blocking technique is an essential optimization parameter. 
Although small blocks are supposed to enable the cache to hold all blocks,
the performance table in~\secref{evaluation:unotpmize on various block sizes} defies our prediction.
Here we see additional experiments and think about why the performance on small blocks is not good.
Since we have already seen the coding performance of $P_{\text{enc}}$ on various blocks,
we first see the performance of the uncompressed but fused version $P^{+F}_{\text{enc}}$ of $P_{\text{enc}}$to check whether or not a similar tendency appears. 
\paragraph{\textbf{Case1: Uncompressed but Fused SLP}}
The following table is the coding throughputs (GB/sec) of $P^{+F}_{\text{enc}}$:
\begin{center}
\begin{tabular}{c|cccccccc}
\text{\small Block size (byte)} & 64 & 128 & 256 & 512 & 1K & 2K & 4K \\\hline
 \textbf{intel} & 0.87 & 1.73 & 2.85 & 4.08 & 5.29 & 5.78 & 4.36 \\
 \textbf{amd} & 1.32 & 2.18 & 3.15 & 3.54 & 3.97 & 4.16 & 3.82
\end{tabular}
\end{center}
where $\textsc{NVar}(P^{+F}_{\text{enc}}) = 32$ and $\cachecap(P^{+F}_{\text{enc}}) = 88$.

We see there are the same patterns at \textbf{intel} and \textbf{amd};
i.e., $2K > 1K > 4K > 512 > 256 > 128 > 64$.
This result again defies our prediction since we need $\mathcal{B} \leq 512$ to avoid cache reloading.

\paragraph{\textbf{Possible reasons for poor performance of small blocks}}
The performance problem on small blocks may cause from two sources.

\paragraph{Cache conflicts in cache sets}
The first source is cache conflicts in \emph{cache sets}, and it prevents cache from holding $32K/\mathcal{B}$ blocks.
Generally, the 32K bytes cache with 8 cache associativity has $\frac{32K}{8} = 4K$ cache sets
where each cache set can hold 8 cache blocks.
Accessing a cache block $b$ whose start address is $\mathcal{A}(b)$,
CPU tries to assign $b$ to the $(\mathcal{A}(b) \bmod 4K)$-th cache set.
If the cache set is full (i.e., it already has 8 cache blocks), CPU evicts the LRU cache block in the set to memory.
Therefore, accessing two blocks $b_1, b_2$ such that $\mathcal{A}(b_1) \equiv_{4K} \mathcal{A}(b_2)$ may cause an eviction in a cache set.

If we take the $4K$-alignment strategy (i.e., locate all blocks on addresses divisible by $4K$),
the cache holds at most 8 blocks regardless of the size $\mathcal{B}$.
To avoid the worst situation, several approaches have been proposed~\cite{Lam:1991, Panda:1999};
however, optimally aligning blocks is a hard problem.
We use a simple approach as follows: for an SLP whose constants are $c_0, c_1, \ldots$,
we allocate $c_i$ so that $\mathcal{A}(c_i) \equiv_{4K} (i \cdot \mathcal{B})$. We do the same for variables $v_1, v_2, \ldots$.
For example, when $\mathcal{B} = 1K$,
\[
\begin{array}{l}
\mathcal{A}(c_0) \equiv_{4K} 0,\ 
\mathcal{A}(c_1) \equiv_{4K} 1K,\ 
\mathcal{A}(c_2) \equiv_{4K} 2K,\ 
\mathcal{A}(c_3) \equiv_{4K} 3K, \\
\mathcal{A}(c_4) \equiv_{4K} 0,\ 
\mathcal{A}(c_5) \equiv_{4K} 1K,\  \ldots
\end{array}
\]
This strategy is better than 4K-alignment since accessing to $c_i$ and $c_j$ never conflict when $ i \not\equiv_{4} j$.
In conclusion, the smaller the block size, the more difficult using cache efficiently as expected.

\paragraph{Latency Penalty}
The second source of the poor performance may be memory access latency.
It is clear that, if $\mathcal{B}$ becomes smaller, then the number of required iteration becomes larger.
Therefore, on small blocks, there are many unavoidable block loading caused by changing iterations.

We now focus the memory access latency on modern CPUs.
For example, we consider Intel's Haswell microarchitecture~\cite{Hammarlund:2014} released in 2013,
and it and its successor are widely used today.
Haswell needs about 150 CPU cycles as latency to reach RAM~\cite{IntelOptim, Hennessy:2017}.
Even if the CPU pipeline maximally works, we need $150 + \frac{\mathcal{B}}{n \cdot 8}$ cycles to load or store a block on an $n$-channel memory.
Thus, we need 158-cycles to load a 64-bytes block at once on a single channel memory.
If we load a 64-bytes block in eight 8-byte loads, then we need $(151 \times 8)$-cycles.
This is the reason why we should load a block from the memory as possible as large.
Haswell can load two 32-bytes data from the cache, XOR the two data using AVX2 or AVX512, and store the result 32-bytes to the L1 cache in \emph{a single cycle}~\cite{Gonsalez:2010}.
Thus, we can perform XORing for two blocks of $\mathcal{B}$ bytes in the cache in $\frac{\mathcal{B}}{32}$ cycles.
When $\mathcal{B} = 64$ (resp. $\mathcal{B} = 4K$), we can perform \textsf{xor32} ${\sim}75$ (resp. 1) times while loading one block from the memory.
In a conclusion, the smaller the block size, then more block loads are required, and the total latency penalty of a small block is relatively larger than that of a large block.

\paragraph{\textbf{Case2: Full Optimization}}
As we have seen above, small blocks may degrade the performance of blocked programs.
Here we see the coding performance of fully optimized---compressed, fused, and scheduled---version of $P_{\text{enc}}$, $P^{\text{Full}}_{\text{enc}}$, to check whether or not large blocks better for blocked programs than smaller ones.
\begin{center}
\scalebox{0.95}{
\begin{tabular}{c|cccccccc}
Block size & 64 & 128 & 256 & 512 & 1K & 2K & 4K \\\hline
 \textbf{intel} (greedy) & 2.29 & 4.00 & 6.02 & 7.61 & 8.68 & 8.37 & 7.24 \\\hdashline
 \textbf{intel} (dfs) & 2.32 & 3.97 & 6.09 & 7.37 & \textbf{8.92} & 8.55 & 7.64 \\\hline
 \textbf{amd} (greedy) & 1.91 & 3.30 & 4.36 & 5.07 & 6.08 & 7.32 & 7.15 \\\hdashline
 \textbf{amd} (dfs) &  1.84 & 3.25 & 4.60 & 5.04 & 6.36 & \textbf{7.58} &  7.31 \\
\end{tabular}
}
\end{center}
where $\mathsf{NVar}(P^{\text{Full}}_{\text{enc}}) \sim 90$ and $\mathsf{CCap}(P^{\text{Full}}_{\text{enc}}) \sim 170$ for all the entries.
We should note that our greedy scheduling heuristic generates different programs for each $\mathcal{B}$.
However, for all $\mathcal{B}$, $\mathsf{NVar}(\cdot)$ is about 90, and $\mathsf{CCap}(\cdot)$ is about 170.
The same is true for the DFS heuristics.

On the basis of the performance,
hereafter we set $\mathcal{B} = 1K$ on \textbf{intel}
and $\mathcal{B} = 2K$ on \textbf{amd} and use the DFS-based scheduling heuristics
for comparison with ISA-L and the previous work.

We consider a reason why the scores of $1K$ and $2K$ are better than that of $4K$ in \textbf{intel}.
Even if conflicts in cache sets happen, the cache with $1K$ and $2K$ blocks may hold more blocks than with $4K$; therefore, the CPU can efficiently use the cache in the case $1K$ and $2K$.
On the other hand, in \textbf{amd}, the score of $1K$ is lower than $2K$ and $4K$.
It possibly comes from a feature of the microarchitecture, \emph{Zen+}, of \textbf{amd}’s CPU.
Zen+, unlike \textbf{intel}’s CPU, performs 256 bitwidth instructions of AVX2, splitting it into two 128 bitwidth instructions~\cite{Fog:2021}.
To put it simply, the performance for AVX2 of \textbf{amd} is half that of \textbf{intel}.
Therefore, in the $1K$ case of \textbf{amd}, we think that the total latency penalty is more significant than the cache efficiency. 

\subsection{Throughput Analysis}\label{evaluation:analyze throughput}
Beyond the average analysis, we optimize the encoding SLP $P_{\text{enc}}$.
\[
\scalebox{0.95}{$
\begin{array}{c|cccc}
 &
 P_{\text{enc}} &
 \textsc{Co}(P_{\text{enc}}) &
 \textsc{Fu}(\textsc{Co}(P_{\text{enc}})) &
 \textsc{Dfs}(\textsc{Fu}(\textsc{Co}(P_{\text{enc}}))) \\\hline
\#_\xor(\cdot) & 755 & 385 & 146 & \leftarrow \\\hdashline
\#_M(\cdot) & 2265 & 1155 & 677 & \leftarrow \\\hdashline
\mathsf{NVar} & 32 & 385 & 146 & 88 \\\hdashline
\mathsf{CCap} & 92 & 447 & 224 & 167 \\\hline\hline
\textbf{intel}(1K) & 4.03 & 4.36 & 7.50 & \textbf{8.92} \\\hline
\textbf{amd}(2K) & 3.17 & 4.46 & 6.62 & \textbf{7.58} \\
\end{array}
$}
\]
where we note that our scheduling heuristics do not affect the number of XORs and memory accesses, and we represent it by $\leftarrow$.

To maximize performance, we can say that compression, fusion, and scheduling are all necessary.
Comparing the first and second columns,
we tell that the number of memory accesses is more dominant than that of $\cachecap$ on the performance.
On the other hand, comparing the third and fourth columns,
we tell that $\cachecap$ certainly represents the cache efficiency.

We also measure the performance of unoptimized and optimized versions of decoding SLPs.
Here, we consider the decoding SLP $P_{\text{dec}}$ obtained by removing $\set{2, 4, 5, 6}$ rows from the encoding matrix
because this SLP has the most XORs---1368 as we see in the following table---among decoding SLPs.
The following table summarizes the related numbers and decoding performance of $P_{\text{dec}}$:
\[
\scalebox{0.95}{$
\begin{array}{c|cccc}
 & 
 P_{\text{dec}} &
 \textsc{Co}(P_{\text{dec}}) &
 \textsc{Fu}(\textsc{Co}(P_{\text{dec}})) &
 \textsc{Dfs}(\textsc{Fu}(\textsc{Co}(P_{\text{dec}}))) \\\hline
\#_\xor & 1368 & 511 & 206 & \leftarrow \\\hdashline
\#_M & 4104 & 1533 & 923 & \leftarrow \\\hdashline
\mathsf{NVar} & 32 & 511 & 206 & 125 \\\hdashline
\mathsf{CCap} & 89 & 585 & 283 & 205 \\\hline\hline
\textbf{intel}(1K) & 2.35  & 3.32 & 5.51 & 6.67 \\\hline
\textbf{amd}(2K) & 2.28  & 3.58 & 5.27 &  6.01 \\
\end{array}
$}
\]
Since $P_{\text{dec}}$ has more instructions than $P_{\text{enc}}$,
we can see that the throughputs of $P_{\text{dec}}$ is smaller than those of $P_{\text{enc}}$.
On the other hand, the overall trend is consistent with $P_{\text{enc}}$.

\begin{figure}
\scalebox{0.9}{
$
\begin{array}{l|c@{\,}c|c@{\,}c|c@{\,}c|c@{\,}c|}
& \multicolumn{2}{|c|}{\#_\xor} & \multicolumn{2}{|c|}{\#_M} & \multicolumn{2}{|c|}{\nvar} & \multicolumn{2}{|c|}{\cachecap} \\
& \text{\ \ Enc} & \text{Dec\ \ } & \text{\ Enc} & \text{Dec\ } & \text{Enc} & \text{Dec} & \text{Enc} & \text{Dec}  \\\hline
\textbf{RS}(8, 4)  & 121 & 170 & 543 & 747 & 79 & 102 & 143 & 166 \\
\textbf{RS}(9, 4)  & 132 & 182 & 611 & 829 & 83 & 117 & 155 & 189 \\
\textbf{RS}(10, 4) & 146 & 206 & 677 & 923 & 88 & 125 & 167 & 205 \\\hline
\textbf{RS}(8, 3)  & 75 & 129 & 364 & 561 & 45 & 77 & 109 & 141 \\
\textbf{RS}(9, 3)  & 87 & 144 & 417 & 641 & 58 & 91 & 128 & 163 \\
\textbf{RS}(10, 3)& 96 & 145 & 471 & 661 & 69 & 85 & 148 & 165 \\\hline
\textbf{RS}(8, 2)  & 26 & 65 & 180 & 286 & 17 & 38 & 80 & 102 \\
\textbf{RS}(9, 2)  & 29 & 73 & 202 & 322 & 19 & 42 & 90 & 113 \\
\textbf{RS}(10, 2) & 30  & 77 & 222 & 352 & 19 & 50 & 98 & 130 
\end{array}
$}
\caption{Values of $\#_\xor, \#_M, \nvar(\cdot)$, and $\cachecap(\cdot)$, of optimized coding SLPs for various codec}
\Description{Values of $\#_\xor, \#_M, \nvar(\cdot)$, and $\cachecap(\cdot)$, of optimized coding SLPs for various codec}
\label{fig:rsvalues}
\end{figure}

\subsection{Throughput Comparison}\label{evaluation:throughput comparison}
We compare the performance of our fully optimized versions of $P_{\text{enc}}$ and $P_{\text{dec}}$ with ISA-L v2.30.0~\cite{ISALv230} and values in~\cite{Zhou:2020}.
As the same as~\cite{Zhou:2020}, we consider three kinds of codec;
4-parities $\textbf{RS}(d, 4)$, 3-parities $\textbf{RS}(d, 3)$, and 2-parities $\textbf{RS}(d, 2)$.
We summarize main measures for each codec in Figure~\ref{fig:rsvalues}.
These values correspond to the rightmost value of the above tables in~\secref{evaluation:analyze throughput}.

We compare the coding throughputs of $\textbf{RS}(d, 4)$ on \textbf{intel} where we use $\mathcal{B} = 1K$ as our blocksize:
\[
\scalebox{0.95}{$
\begin{array}{c|cc|cc|cc|}
\multirow{2}{*}{\begin{tabular}{c}\textbf{intel 1K}\\\text{(GB/sec)}\end{tabular}}
& \multicolumn{2}{|c|}{\text{Ours}} & \multicolumn{2}{|c|}{\text{ISA-L v2.30}} & \multicolumn{2}{|c|}{\text{Values~of~\cite{Zhou:2020}}} \\
& \text{Enc} & \text{Dec} & \text{Enc} & \text{Dec} & \text{Enc} & \text{Dec} \\\hline
\textbf{RS}(8, 4) & 8.86 & 6.78 & 7.18 & 7.04 & 4.94 & 4.50 \\
\textbf{RS}(9, 4) & 8.83 & 6.71 & 6.91 & 6.58 & \multicolumn{2}{|c|}{\text{\small Not Available in~\cite{Zhou:2020}}} \\
\textbf{RS}(10, 4) & 8.92 & 6.67 & 6.79 & 4.88 & 4.94 & 4.71
\end{array}
$}
\]
The table claims that our EC library exceeds ISA-L in encoding and parallels in decoding.

Let us consider why there is no difference in the performance between \textbf{RS}(10, 4) and \textbf{RS}(8, 4), although \textbf{RS}(8, 4) is better than \textbf{RS}(10, 4) in terms of the measures in Figure~\ref{fig:rsvalues}.
To encode or decode a given data of $N$-bytes,
we run SLPs for $8 \times 8$ input arrays of $\frac{N}{8 \times 8}$-bytes on \textbf{RS}(8, 4).
Similarly, we run SLPs for $8 \times 10$ input arrays of $\frac{N}{8 \times 10}$-bytes on \textbf{RS}(10, 4).
The difference in Figure~\ref{fig:rsvalues} is due to the fact that the number of input arrays of \textbf{RS}(10, 4) is larger than that of \textbf{RS}(8, 4).
On the other hand, the total number of iterations $\frac{N}{8 \times 8 \times \mathcal{B}}$ on \textbf{RS}(8, 4) is larger than that $\frac{N}{8 \times 10 \times \mathcal{B}}$ on \textbf{RS}(10, 4). As a result, there is no difference in the performance between them.

The encoding and decoding performance of ISA-L are close; however, there is indeed difference between them in our EC library.
Our encoding and decoding matrices, which are sources of $P_{\text{enc}}$ and $P_{\text{dec}}$, equal those of ISA-L in the binary representation. Namely, we use the same matrix that ISA-L uses.
We believe that this situation could be caused by the following fact.
In ISA-L or EC libraries based on MM over finite fields algorithms,
for coding matrices $M_1$ and $M_2$ and a data matrix $D$,
there is not much difference between the two computational costs of $M_1 \cdot D$ and $M_2 \cdot D$
because finite field multiplication is usually implemented using a multiplication table $\mathcal{M}$; i.e., $a \cdot b$ is computed by accessing $\mathcal{M}[a][b]$.
On the other hand, in our libraries or EC libraries based on the method XOR-based EC,
even if the sizes of two matrices $M_1$ and $M_2$ are equal,
the number of required XORs could be very different.
For example, for an element $e_1, e_2 \in \gf{2^8}$, the number of 1 in the bitmatrix $\tilde{e_1}$ may be significantly larger than those of $\tilde{e_2}$.
In a conclusion, we consider the performance gap between encoding and decoding in our library is intrinsic in XOR-based EC.

We also have similar structures in the throughputs table of $\textbf{RS}(d, 4)$ on \textbf{amd} where we use $\mathcal{B} = 2K$ as the blocksize:
\[
\scalebox{0.95}{$
\begin{array}{c|cc|cc|cc|}
\multirow{2}{*}{\begin{tabular}{c}\textbf{amd 2K}\\\text{(GB/sec)}\end{tabular}}
& \multicolumn{2}{|c|}{\text{Ours}} & \multicolumn{2}{|c|}{\text{ISA-L v2.30}} & \multicolumn{2}{|c|}{\text{Values~of~\cite{Zhou:2020}}} \\
& \text{Enc} & \text{Dec} & \text{Enc} & \text{Dec} & \text{Enc} & \text{Dec} \\\hline
\textbf{RS}(8, 4) & 7.09 & 5.53 & 4.60 & 4.61 & 4.69 & 4.06 \\
\textbf{RS}(9, 4) & 7.22 & 5.86 & 4.70 & 4.70 & \multicolumn{2}{|c|}{\text{\small Not Available in~\cite{Zhou:2020}}} \\
\textbf{RS}(10, 4) & 7.58 & 6.01 & 4.76 & 4.75 & 4.67 & 3.91
\end{array}
$}
\]

\paragraph{Comparison in Low Parities}
We compare \textbf{RS}($d$, 3) and \textbf{RS}($d$, 2):
\[
\scalebox{0.9}{$
\begin{array}{c|cc|cc|ll|}
\multirow{2}{*}{\begin{tabular}{c}\textbf{intel 1K}\\\text{(GB/sec)}\end{tabular}}
& \multicolumn{2}{|c|}{\text{Ours}} & \multicolumn{2}{|c|}{\text{ISA-L v 2.30}} & \multicolumn{2}{|c|}{\text{Values~of~\cite{Zhou:2020}}} \\
& \text{Enc} & \text{Dec} & \text{Enc} & \text{Dec} & \text{Enc} & \text{Dec} \\\hline
\textbf{RS}(8, 3) & 12.32 & 8.82 & 9.09 & 9.25 & 6.08 & 5.57 \\
\textbf{RS}(9, 3) & 11.97 & 8.27 & 7.31 & 7.92 & 6.17 & 5.66 \\
\textbf{RS}(10, 3) & 11.78 & 8.89 & 6.78 & 7.93 & 6.15_{S} & 5.90 \\\hline\hline
\textbf{RS}(8, 2) & 18.79 & 14.59 & 12.99 & 13.34 & 8.13_{E} & 8.07_{E} \\
\textbf{RS}(9, 2) & 18.93 & 14.27 & 11.85 & 12.03 & 8.34_{E} &  8.04 \\
\textbf{RS}(10, 2) & 18.98 & 14.66 & 12.12 & 12.61 & 8.40_{E} & 8.22_{E} \\
\end{array}$
}
\]
\[
\scalebox{0.95}{$
\begin{array}{c|cc|cc|ll|}
\multirow{2}{*}{\begin{tabular}{c}\textbf{amd 2K}\\\text{(GB/sec)}\end{tabular}}
& \multicolumn{2}{|c|}{\text{Ours}} & \multicolumn{2}{|c|}{\text{ISA-L v 2.30}} & \multicolumn{2}{|c|}{\text{Values~of~\cite{Zhou:2020}}} \\
& \text{Enc} & \text{Dec} & \text{Enc} & \text{Dec} & \text{Enc} & \text{Dec} \\\hline
\textbf{RS}(8, 3) & 9.35 & 7.43 & 5.01 & 4.93 & 6.38_{S} & 5.18_{Q} \\
\textbf{RS}(9, 3) & 9.41 & 7.44 & 5.07 & 5.02 & 6.53_{S} & 6.53_{S} \\
\textbf{RS}(10, 3) & 9.51 & 7.46 & 5.04 & 5.02 & 6.49_{S} & 5.31_{Q} \\\hline\hline
\textbf{RS}(8, 2) & 13.60 & 12.07 & 7.11 & 7.09 & 8.96_{R} & 10.11_{E} \\
\textbf{RS}(9, 2) & 13.83 & 12.06 & 7.17 & 7.19 & 9.12_{R} & 9.31_{R} \\
\textbf{RS}(10, 2) & 14.13 & 12.19 & 7.24 & 7.15 & 9.31_{R} & 10.60_{R} \\
\end{array}
$}
\]
The column "Values of~\cite{Zhou:2020}" consists of the best throughputs among results of the corresponding parameters in~\cite{Zhou:2020} where the authors compared their proposal method with some codecs specialized for low parities---STAR~\cite{Huang:2008} and QFS~\cite{Ovsiannikov:2013} for three parities, and EvenOdd~\cite{Blaum:1995} and RDP~\cite{Corbett:2004} for two parities.
Indeed, the values $\cdot_{S}$, $\cdot_{Q}$, $\cdot_{E}$, and $\cdot_{R}$ are scored
by STAR, QFS, EvenOdd, and RDP, respectively (the other values are scored by their proposal approach).
We can say our library works well without specializing for low parities.

\section{Conclusion and Future Work}
We have proposed a streamlined approach to implement an efficient XOR-based EC library.
We combined the four notions,
straight-line programs (SLPs) from program optimization,
grammar compression algorithm \textsc{RePair},
the functional program optimization technique deforestation,
and the pebble game from program analysis.
We extended \textsc{RePair} to our \textsc{XorRePair} to accommodate the cancellative property of XOR.
We used the pebble game to model SLPs with the abstract LRU cache.
Orthogonally composing these methods, we have implemented an experimental library that outperforms Intel's high-performance library, ISA-L~\cite{ISAL}.

Analyzing the result of experiments, we have noticed the importance of cache optimization.
In this paper, we only tried to abstract the L1 cache but not the L2 and L3 caches.
We are thinking about using the multilevel pebble game introduced by Savage in~\cite{Savage:1995} to accommodate the L2 and L3.
As a related cache efficiency topic, we are interested in automatically inserting software prefetches~\cite{Lee:2012}.
It may hide the cache transfer penalty from memory to cache if a CPU concentrates on performing array XORs against cached data.

\section*{Acknowledgments}
We gratefully thank anonymous reviewers for their invaluable and thorough comments, which improved the presentation of this paper and also helped us improve the performance of our experimental library.
Many thanks to our colleague Masahiro Fukasawa for fruitful discussions of cache optimization.
Thanks also to Iori Yoneji for his full support in providing evaluation environments.

\bibliographystyle{ACM-Reference-Format}
\bibliography{mybib}

\appendix


\newcommand{\pnt}{\Gamma}

\setcounter{lemma}{0}
\renewcommand{\thelemma}{\Alph{section}\arabic{lemma}}
\setcounter{proposition}{0}
\renewcommand{\theproposition}{\Alph{section}\arabic{proposition}}

\section{Proof of Theorem~\ref{thm:intractability of minimum memory access problem}}\label{sec:appendixA}

In this appendix section, we show the following problem of \secref{sec:Fusion} is NP-complete.
\namedframe{\textbf{The minimum memory access problem}}{
For an $P \in \slp_\xor$,
we find $Q \in \multislp$ that satisfies $\result{P} = \result{Q}$ and minimizes $\#_M(Q)$.
}

\subsection{Vertex Cover Problem}
To show the NP-completeness of our problem,
we use the classic NP-complete Vertex Cover problem~\cite{Garey:1990}.

The vertex cover problem (VCP) is a decision problem such that:
\begin{itemize}
\item Let $G$ be an undirected graph $G$ and $k$ be a natural number.
\item We then decide if there is a node set $X$ of $G$ such that
\begin{itemize}
\item $|X| \leq k$; and
\item $G$ is covered by $X$. Namely, every edge $(a, b)$ of $G$ is covered by $X$; i.e., $a \in X$ or $b \in X$ holds.
\end{itemize}
\end{itemize}

The optimization version of VCP, OptVCP, is an optimization problem such that:
\begin{itemize}
\item Let $G$ be an undirected graph.
\item We then compute the smallest node set $X$ of $G$ that covers $G$.
\end{itemize}
Since VCP is NP-complete, OptVCP cannot be solved in polynomial time unless \textbf{P=NP}.

Below we reduce our optimization problem, the minimum memory access problem, to OptVCP.

\subsection{Build SLP from Graph}
Let $G$ be an undirected graph.

Before building an SLP corresponding to $G$,
we modify $G$ for our construction and proof as follows:
\begin{itemize}
\item For each node $a$ of $G$, we add two fresh nodes $\lambda_a$ and $\mu_a$
and add two edges $(a, \lambda_a)$ and $(a, \mu_a)$.
\begin{itemize}
\item We call added nodes $\lambda_\bullet$ and $\mu_\bullet$ \emph{local nodes}.
\item For a node $a$ of $G$, we call edges $(a, \lambda_a)$ and $(a, \mu_a)$ \emph{local edges}.
\end{itemize}
\end{itemize}
Hereafter we omit the subscripts of $\lambda_\bullet$ and $\mu_\bullet$ if they are clear from the context.
For example, we can simply write $(a, \lambda)$ because there is no edge such that $(b, \lambda_a)$.

For the modified graph, we build an SLP as follows:
\begin{itemize}
\item For each edge $(x, y) \in G_{\text{modif}}$, we add a goal $g_{x y} \gets \rho \xor x \xor y$.
\end{itemize}

Let us consider the following example:
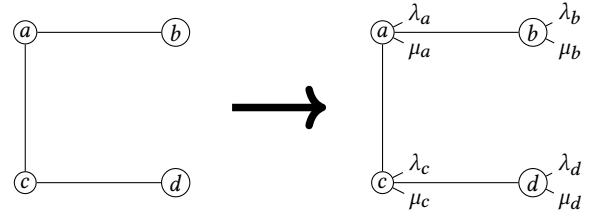
\begingroup
\tikzstyle{none}=[]
\tikzstyle{vertex}=[fill=none, draw=black, shape=circle, inner sep=1pt]
\tikzstyle{localv}=[fill=none, draw=none, shape=circle, dashed, inner sep=1pt]

\tikzstyle{new edge style 0}=[->]
\tikzstyle{new edge style 1}=[fill=black, ->, line width=1mm]

\begin{tikzpicture}
	\begin{pgfonlayer}{nodelayer}
		\node [style=vertex] (0) at (-4.25, 1) {$a$};
		\node [style=vertex] (1) at (-4.25, -1) {$c$};
		\node [style=vertex] (2) at (-2.25, 1) {$b$};
		\node [style=vertex] (3) at (-2.25, -1) {$d$};
		\node [style=vertex] (4) at (0.5, 1) {$a$};
		\node [style=vertex] (5) at (0.5, -1) {$c$};
		\node [style=vertex] (6) at (2.5, 1) {$b$};
		\node [style=vertex] (7) at (2.5, -1) {$d$};
		\node [style=localv] (8) at (1, 1.25) {$\lambda_a$};
		\node [style=localv] (9) at (1, 0.75) {$\mu_a$};
		\node [style=localv] (10) at (3, 1.25) {$\lambda_b$};
		\node [style=localv] (11) at (3, 0.75) {$\mu_b$};
		\node [style=localv] (12) at (1, -0.75) {$\lambda_c$};
		\node [style=localv] (13) at (1, -1.25) {$\mu_c$};
		\node [style=localv] (14) at (3, -0.75) {$\lambda_d$};
		\node [style=localv] (15) at (3, -1.25) {$\mu_d$};
		\node [style=none] (16) at (-1.5, 0) {};
		\node [style=none] (17) at (-0.25, 0) {};
	\end{pgfonlayer}
	\begin{pgfonlayer}{edgelayer}
		\draw (0) to (2);
		\draw (0) to (1);
		\draw (1) to (3);
		\draw (4) to (6);
		\draw (4) to (5);
		\draw (5) to (7);
		\draw (4) to (8);
		\draw (4) to (9);
		\draw (6) to (10);
		\draw (6) to (11);
		\draw (5) to (12);
		\draw (5) to (13);
		\draw (7) to (14);
		\draw (7) to (15);
		\draw [style=new edge style 1] (16.center) to (17.center);
	\end{pgfonlayer}
\end{tikzpicture}

\endgroup

The left graph $G$ is modified to the right graph $G_{\text{modif}}$.
We build the SLP $P_G$ from $G_{\text{modif}}$ as follows:
\[
\begin{array}{l@{\ }l}
g_{a b} & \gets \rho \xor a \xor b; \\
g_{a c} & \gets \rho \xor a \xor c; \\
g_{c d} & \gets \rho \xor c \xor d; \\
g_{a \lambda_a} & \gets \rho \xor a \xor \lambda_a; \\
g_{a \mu_a} & \gets \rho \xor a \xor \mu_a; \\
\multicolumn{2}{l}{\phantom{aaaaa}\vdots} \\
g_{d \lambda_d} & \gets \rho \xor a \xor \lambda_d; \\
g_{d \mu_d} & \gets \rho \xor a \xor \mu_d; \\
\multicolumn{2}{c}{\return(g_{a b}, g_{a c}, g_{c d}, g_{a \lambda_a}, \ldots, g_{d \mu_d});}
\end{array}
\]

We will show that we can extract the minimum cover sets of $G$ (rather than $G_{\text{modif}}$)
from the solution of the minimum memory access problem for $P_G$.
For example, the following $\multislp$ is one of the minimum solution:
\[
\begin{array}{l}
\pnt_a \gets \rho \xor a; \\
g_{a b} \gets \pnt_a \xor b; \quad
g_{a c} \gets \pnt_a \xor c; \quad
g_{a \lambda} \gets \pnt_a \xor \lambda; \quad
g_{a \mu} \gets \pnt_a \xor \mu; \\[5pt]
g_{b \lambda} \gets \rho \xor b \xor \lambda; \quad
g_{b \mu} \gets \rho \xor b \xor \mu; \\[5pt]
\pnt_c \gets \rho \xor c; \\
g_{c d} \gets \pnt_c \xor d; \quad
g_{c \lambda} \gets \pnt_c \xor \lambda; \quad
g_{c \mu} \gets \pnt_c \xor \mu; \\[5pt]
g_{d \lambda} \gets \rho \xor d \xor \lambda; \quad
g_{d \mu} \gets \rho \xor d \xor \mu; \\
\return(\cdots)
\end{array}
\]
From the solution, we extract $a$ and $c$ as cover sets of the original graph $G$; actually, $\set{a, c}$ covers $G$.

Technically, we prove the following two lemmas in the subsequent sections.
\begin{lemma}\label{appendix main lemma}
Let $Q$ be an $\multislp$ where $\result{P_G} = \result{Q}$.

We can effectively normalize $Q$ to $Q'$ such that;
\begin{itemize}
\item $\#_M(Q') \leq \#_M(Q)$.
\item Every edge $(a, b)$ of $G$ is represented in $Q'$ as follows:
\[
\begin{array}{l}
\pnt_a \gets \rho \xor a; \\
g_{a b} \gets \pnt_{a} \xor b
\end{array}
\]
\item For a node $a$ of $G$, the local edges $(a, \lambda)$ and $(a, \mu)$ are represented in $Q'$ as follows:
\[
\left(
\begin{array}{l}
\text{if $\pnt_a$ is in $Q'$} \\\hline
g_{a \lambda} \gets \pnt_a \xor \lambda; \\
g_{a \mu} \gets \pnt_a \xor \mu.
\end{array}
\right)
\quad\text{OR}\quad
\left(
\begin{array}{l}
\text{if $\pnt_a$ is not in $Q'$} \\\hline
g_{a \lambda} \gets \rho \xor a \xor \lambda; \\
g_{a \mu} \gets \rho \xor a \xor \mu.
\end{array}
\right)
\]
\end{itemize}
\end{lemma}

Using this lemma, we obtain the following two useful properties.
\begin{lemma}\label{appendix:size lemma}
Let $P$ be a normalized $\multislp$ in the meaning of Lemma~\ref{appendix main lemma}.
If the size of the set $\set{ a \in G : \text{$\pnt_a$ is in $Q'$} }$ is $k$,
then
\[
3|E| + 8|N| + k = \#_M(Q').
\]
where $E$ is the edge sets of $G$ and $N$ is the node sets of $G$.
\end{lemma}
\begin{proof}
For each temporal variables, we need 3-costs for defining $\pnt_a \gets \rho \xor a$.
This totally require $3k$-costs.

For each local edge of a node $a$,
\begin{itemize}
\item if we have $\pnt_a$, we require 6-costs: $g_{a \lambda} \gets \pnt_a \xor \lambda$ and
$g_{a \mu} \gets \pnt_a \xor \mu$.
\item Otherwise, we require 8-costs: $g_{a \lambda} \gets \rho \xor a \xor \lambda$ and
$g_{a \mu} \gets \rho \xor a \xor \mu$.
\end{itemize}

Totally, for defining all the local edges, we require $6k + 8(|N| - k)$-costs.

For each edge $(a, b)$ of $G$,
since we have $\pnt_a$ or $\pnt_b$,
we require 3-costs $g_{a b} \gets \pnt_a \xor b$ or $g_{a b} \gets \pnt_b \xor a$.

Entirely, we need the following costs matching with one of the statement:
\[
3k + (6k + 8(|N| - k)) + 3|E| =k + 8|N| + 3|E|.
\]
\end{proof}

Using these lemmas, we obtain the following property.
\begin{lemma}\label{appendix:VCP-MMAP}
For a given graph $G$, let $Q$ be the $\#_M$-minimum SLP such that $\result{Q} = \result{P_G}$.
The following holds on $Q$:
\begin{itemize}
\item There is a number $k$ such that $3|E| + 8|N| + k = \#_M(Q)$.
\item Furthermore, $k$ is the size of the smallest cover set of $G$.
\end{itemize}
\end{lemma}
\begin{proof}
We normalize $Q$ to $Q'$ using Lemma~\ref{appendix main lemma} and then estimate $k$ using Lemma~\ref{appendix:size lemma}.
The conditions of Lemma~\ref{appendix main lemma} tells we can cover the graph $G$ using $k$ vertices.

Let $X = \set{x_1, x_2, \ldots, x_n}$ be the smallest cover set of $G$.
We then can construct an $\multislp$ $R$
such that $R$ satisfies the conditions of Lemma~\ref{appendix main lemma} and $\result{R} = \result{P_G}$
defining variables $\pnt_{x_1}, \ldots, \pnt_{x_n}$
as
\[
\pnt_{x_1} \gets \rho \xor x_1;\ 
\pnt_{x_2} \gets \rho \xor x_2;\ 
\ \ \ldots\ \ 
\pnt_{x_n} \gets \rho \xor x_n;\ 
\]
Lemma~\ref{appendix:size lemma} tells $\#_M(R) = 3|E| + 8|N| + n$.

From the $\#_M$-minimality of $Q$, $k \leq n$ must hold.
On the other hand, from the minimality of $n$, $n \leq k$ must hold; then, we have $k = n$.
\end{proof}

Lemma~\ref{appendix:VCP-MMAP} immediately leads to the intractability of our optimization problem.

Hereafter, we will show Lemma~\ref{appendix main lemma}.

\subsection{Proof of Lemma~\ref{appendix main lemma}}

\paragraph{\textbf{Terminology}}
The following is terminology for this section:
\begin{description}
\item[Temporal variable:] If a variable is not a goal, we call it \emph{temporal} variable.
\item[Goal variable:] If a variable is one of the goal, we call it \emph{goal} variable.
\begin{itemize}
\item In our setting, each goal variable is of the form $g_{a b}$, $g_{a \lambda}$, or $g_{a \mu}$
where $a$ and $b$ are nodes and $\lambda$ and $\mu$ are local nodes.
\end{itemize}
\item[Freely appearing:]
If a variable $v$ or a constant $c$ appears in the definition of a temporal variable,
we say that $v$ or $c$ freely appears in the program
\end{description}

\paragraph{\textbf{Metavariable Naming Rule}}
\begin{description}
\item[$\boldsymbol{a, b, \ldots}$:] Metavariables for nodes of $G$.
\item[$\boldsymbol{x, y, z}$:] Metavariables for nodes of $G_{\text{modif}}$.
\item[$\boldsymbol{t, t_1, t_2, \ldots}$:] Metavariables for terms (i.e., constants or variables) of SLPs.
\item[$\boldsymbol{v, v_1, v_2, \ldots}$:] Metavariables for a temporal variable of SLPs.
\end{description}

\paragraph{\textbf{Notation}}
Let $P$ be an SLP and $v$ be a variable of $P$.

We write $\val{v}$ to denote the value of $v$.
To justify this notation, we assume that
every SLP of this section is of the \emph{SSA (single static assignment)} form
where each variable is assigned exactly once~\cite{Alpern:1988, Rosen:1988}.
We can easily convert a given SLP to an SSA form SLP without changing the semantics and the size.
For example,
\[
\begin{array}{l}
\hfill \text{non-SSA SLP} \hfill \\\hline
v \gets a \xor b \xor c; \\
v \gets v \xor d \xor e; \\
v \gets v \xor f;
\end{array}
\qquad \Mapsto \qquad
\begin{array}{l}
\hfill \text{SSA SLP} \hfill \\\hline
v_1 \gets a \xor b \xor c; \\
v_2 \gets v_1 \xor d \xor e; \\
v_3 \gets v_2 \xor f;
\end{array}
\]

\subsubsection{Normalization I}\mbox{}\\
Let $v$ be a temporal variable.

If $\val{v} = \set{ x, y }$ and $v \gets t_1 \xor t_2 \xor \cdots$,
we convert the definition to the trivial form as follows:
\[
v \gets t_1 \xor t_2 \xor \cdots;\ \Mapsto v \gets x \xor y;
\]

If $\val{v} = \set{ x, y, z }$ and $v \gets t_1 \xor t_2 \xor t_3 \xor \cdots$,
we convert the definition to the trivial form as follows:
\[
v \gets t_1 \xor t_2 \xor t_3 \xor \cdots;\ \Mapsto v \gets x \xor y \xor z;
\]
These conversion do not change the semantics of the program and increase the size of the program.

\subsubsection{Normalization II}\mbox{}\\
If the program normalized by Normalization-I has a variable $\pnt_a$ of the form:
\[
\pnt_a \gets \rho \xor a;
\]
using $\pnt_a$, we rewrite the definitions of $g_{a \lambda}$ and $g_{a \mu}$ as follows:
\[
g_{a \lambda} \gets \pnt_a \xor \lambda;
\qquad
g_{a \mu} \gets \pnt_a \xor \mu;
\]

The following holds after these normalization steps.
\begin{proposition}\label{prop:freely appear}
If $\pnt_a$ is not in the (normalized) program,
for $\lambda_a$ (also $\mu_a$),
either one of the following holds:
\begin{itemize}
\item $g_{a \lambda}$ is defined as $g_{a \lambda} \gets \rho \xor a \xor \lambda$; or
\item $\lambda$ freely appears.
\end{itemize}
\end{proposition}
\begin{proof}
We assume that the former does not hold
and then consider the following subcases about $g_{a \lambda}$:
\begin{description}
\item[$g_{a \lambda} \gets v \xor \lambda$ where $\val{v} = \set{\rho, a}$:]
This contradicts to the assumption.
\item[$g_{a \lambda} \gets v \xor x$ ($x \neq \lambda_a$) or $g_{a \lambda} \gets v \xor \rho$ where $\val{v} = \set{\lambda, \ldots}$:]
To define $v$, $\lambda$ must freely appear because any other goal does not have $\lambda_a$.
\item[$g_{a \lambda} \gets v \xor g_{x y}$ where $\val{v} = \set{ \lambda, \ldots}$:]
To define $v$, $\lambda$ must freely appear.
\item[$g_{a \lambda} \gets v_1 \xor v_2$ where $\val{v_1} = \set{ \lambda, \ldots}$:]
To define $v_1$, $\lambda$ must freely appear.
\end{description}
\end{proof}

On the basis of this lemma, we introduce one notation.
\begin{description}
\item[Movable occurrence:] For a local node $\lambda$,
if its appears in $g_{a \lambda} \gets \rho \xor a \xor \lambda$
or $\lambda$ freely appears as $v \gets \lambda \xor \cdots$,
we call such occurrences of $\lambda$ \emph{movable occurrences}.
\end{description}
The meaning of the term \emph{movable} will be justified below.

\subsubsection{Normalization III: Deleting Cancellation}\mbox{}\\
Let $P$ be an $\multislp$ normalized by the steps Normalization-$\set{I, II}$.
From $P$, we build an SLP $Q$ where $Q$ does not have the XOR cancellation and $\#_M(Q) \leq \#_M(P)$.

\def\unfold{\ensuremath{\mathit{unfold}}}

We introduce an auxiliary function \textit{unfold} to (recursively) unfold the definition of a \emph{temporal} variable:
\[
\begin{array}{ll}
\unfold(x) = \set{ x } & \text{if $x$ is a constant} \\
\unfold(g_{x y}) = \set{ g_{x y} } \\ 
\unfold(v) = \bigoplus\limits^{n}_{i=1} \mathit{unfold}(t_i) & \text{if $v \gets t_1 \xor t_2 \xor \cdots \xor t_n$}
\end{array}
\]
It should be noted that we do not unfold the definition of a goal for our construction.

Using constants and goals that freely appear in $P$, we define a graph $\Game$ as follows:
\begin{enumerate}
\item First, we define a graph $\Game$ as a complete graph whose nodes are constants that freely appear in $P$.
\item Next, for each freely appearing goal $g_{x y}$ in $P$, we add the edge $(x, y)$ to $\Game$.
\end{enumerate}

We note the obtained graph $\Game$ may not be connected; i.e., it may have multiple components.

\paragraph{\textbf{Notations}}
\newcommand{\gpath}{\rightleftharpoons}
\begin{description}
\item[$\Game_\star$:] We call the unique component of $\Game$ that contains freely appearing constants $\Game_\star$.
\item[$\gpath$:] We write $a \gpath b$ to the path between $a$ and $b$ of $\Game$.
\end{description}

This graph $\Game$ has the following useful properties.
\begin{proposition}\label{appendix:val-path}
Let $v$ be a temporal variable of $P$.
\begin{enumerate}
\item If $\val{v} = \set{x, y}$,
then $x \gpath y$.
\item If $\val{v} = \set{x_1, x_2, x_3, x_4}$, 
then this consists of two paths $n_1 \gpath n_2$ and $n_3 \gpath n_4$
where $\set{x_1, x_2, x_3, x_4} = \set{n_1, n_2, n_3, n_4}$.
\item If $\val{v} = \set{\rho, x, y, z}$,
then this consists of one path $n_1 \gpath n_2$ and one point $n_3 \in \Game_\star$,
where $\set{x, y, z} = \set{n_1, n_2, n_3}$.
\end{enumerate}
\end{proposition}
\begin{proof}
We show (1). The other cases are shown by the same argument.

Let $X = \unfold(v)$.
We reduce $X$ repeatedly applying one of the following replacements rules:
\begin{enumerate}
\item Choose $x \in \Game_\star$ and $y \in \Game_\star$ from $X$, remove them from $X$, then update $X \coloneq X \xor \set{(x, y)}$.
Remark $(x, y) \in \Game$.
\item Choose $x \in \Game_\star$ and $(x, y) \in \Game$ from $X$, remove them from $X$, then update $X \coloneq X \xor \set{y}$.
Remark $y \in \Game_\star$.
\item Choose $(x, y), (y, z)$ from $X$ where $x \gpath y$ and $y \gpath z$, remove them them from $X$, then update $X \coloneq X \xor \set{(x, z)}$.
Remark $(x, z) \in \Game$.
\end{enumerate}
When we cannot apply any rule, then $X = \set{ (a, b) }$; thus, $a \gpath b$.
\end{proof}
\begin{proposition}\label{appendix:vv prop}
If a goal $g_{x y}$ is defined by two temporal variables $v_1, v_2$
(i.e., $g_{x y} \gets v_1 \xor v_2$), then $x \gpath y$.
\end{proposition}
\begin{proof}
By applying the same argument of the above proposition for $X = \unfold(v_1) \xor \unfold(v_2)$,
we have $x \gpath y$.
\end{proof}

Hereafter, we build an $\multislp$ $Q$ that does not use the XOR-cancellation.

First, we copy $P$ to be normalized to $P'$.

\paragraph{Removing nodes that freely appear in $P'$}.

Let $a_1, a_2, \ldots$ are (non-local) nodes that freely appear in $P'$,
\begin{itemize}
\item If $\pnt_{a_i}$ is in $P'$, we move it to $Q$.
\begin{itemize}
\item \emph{Moving} means that we remove $\pnt_{a_i} \gets \rho \xor a_i$ from $P'$ and then
add it to $Q$.
\end{itemize}
\item Otherwise, by Proposition~\ref{prop:freely appear},
we have movable occurrences of $\lambda_{a_i}$ and $\mu_{a_i}$ in $P'$.

We remove all the free occurrences of $\pnt_{a_i}$ and all the movable occurrences of $\lambda_{a_i}$ and $\mu_{a_i}$ from $P'$.
We then add $\pnt_{a_i} \gets \rho \xor a_i$ to $Q$.
\end{itemize}

\paragraph{Removing goal variables that freely appear in $P'$}\mbox{}\\
Let $g_1, g_2, \ldots$ are goal variables that freely appear in $P'$.
\begin{itemize}
\item Let $g_i = g_{a b}$.
\item Let $C$ be the component of $\Game$ that the edge $(a, b)$ belongs to.
\item If there is $x \in C$ such that $\pnt_x$ is not in $Q$,
by Proposition~\ref{prop:freely appear},
we have movable occurrences of $\lambda_x$ and $\mu_x$ in $P'$.

We remove all the free occurrences of $g_i$
and all the movable occurrences of $\lambda_{c}$ and $\mu_{c}$.

We then add $\pnt_c \gets \rho \xor c$ to $Q$.
\end{itemize}

After the above modification to $Q$, the following property holds.
\begin{proposition}\label{prop:one of path}
\mbox{}
\begin{itemize}
\item $\#_M(P') + \#_M(Q) \leq \#_M(P)$. 
\item For each non-local node $a \in \Game_\star$, we have $\pnt_a$ in $Q$.
\item For each path $a \gpath b$, we have $\pnt_a$ or $\pnt_b$ in $Q$.
\item For each path $a \gpath x$ where $x$ is local, we have $\pnt_a$ in $Q$. \\
We do not need $x \in \set{ \lambda_a, \mu_a }$.
\end{itemize}
\end{proposition}

\paragraph{Removing temporal variables of specific patterns in $P'$}\mbox{}\\
Let $v$ be a temporal variable in $P$.
\begin{description}
\item[Pat1:] If $\val{v} = \set{\lambda_a, \rho}$ and $\pnt_a$ is not in $Q$,
there are movable occurrences of $\lambda_a$ and $\mu_a$ in $P'$.
\begin{enumerate}
\item First, we remove all the such occurrences.
\item Next, we remove the left-side occurrence of $v$ of the definition of $v$: i.e.,
\[
v \gets \cdots; \quad \Mapsto \quad \_ \gets \cdots;
\]
\begin{itemize}
\item Although such the blank is not permitted in our formalization of SLP,
we temporarily allow it only on $P'$ because we use $P'$ to build $Q$ and do not execute $P'$.
\item By these removal, we decrement the size of $P'$ more than three.
Recall that the size of $\multislp$ equals to the number of the total occurrences of constants and variables.
\end{itemize}
\item Finally, we add $\pnt_a \gets \rho \xor a$ to $Q$. It increments the size of $Q$ by three.
\end{enumerate}

Totally, $\#_M(P') + \#_M(Q) \leq \#_M(P)$ still holds.
\item[Pat2:] If $\val{v} = \set{\lambda_a, \mu_a}$ and $\pnt_a$ is not in $Q$,
by the same argument of Pat1,
we remove the left side occurrence of $v$ and all the movable occurrences of $\lambda_a$ and $\mu_a$;
then, we add $\pnt_a \gets \rho \xor a $ to $Q$.
\item[Pat3:] If $\val{v} = \set{\lambda_a, b}$ and $\pnt_a$ is not in $Q$,
we remove the definition of $v$ from $P'$ and then add $\pnt_a \gets \rho \xor a$ to $Q$.
\item[Pat4:] If $\val{v} = \set{a, b}$ and $\pnt_a$ (or $\pnt_b$) is not in $Q$,
then we remove the left side occurrence of $v$ and all the movable occurrences of $\lambda_a$ and $\mu_a$
(or $\lambda_b$ and $\mu_b$) from $P'$.

If $\pnt_a$ is not in $Q$, we add $\pnt_a \gets \rho \xor a$ to $Q$.
Otherwise, if $\pnt_b$ is not in $Q$, we add $\pnt_b \gets \rho \xor b$ to $Q$.

Consequently, we have $\pnt_a$ and $\pnt_b$ in $Q$.
\item[Pat5:] If $\val{v} = \set{a, b, c, d}$,
by Proposition~\ref{appendix:val-path}~and~\ref{prop:one of path}
there are at most two nodes $x, y \in \set{a, b, c, d}$ such that $\pnt_x$ and $\pnt_y$ are not in $Q$.

We remove the left side occurrence of $v$ and all the movable occurrences of $\lambda_x$, and $\mu_x$.
We then add $\pnt_x \gets \rho \xor x$ to $Q$.

Consequently, there is at most one node $y \in \set{a, b, c, d}$ such that $\pnt_y$ is not in $Q$.
\item[Pat6:] If $\val{v} = \set{a, \lambda_a, b, \lambda_b}$,
by Proposition~\ref{appendix:val-path}~and~\ref{prop:one of path},
$\pnt_a$ or $\pnt_b$ may not exist in $Q$.

If $\pnt_a$ is not in $Q$, we remove the left side occurrence of $v$
and all the movable occurrences of $\lambda_x$ and $\mu_x$.
We then add $\pnt_a \gets \rho \xor a$ to $Q$.

Otherwise, we do the same for $\pnt_b$.

Consequently, we have $\pnt_a$ and $\pnt_b$ in $Q$.
\item[Pat7:] If $\val{v} = \set{a, \lambda_a, b, c}$,
by Proposition~\ref{appendix:val-path}~and~\ref{prop:one of path},
we may do not have $\pnt_x$ for at most one of $\set{a, b, c}$.

We remove $v$, $\lambda_x$, and $\mu_x$ and then add $\pnt_x \gets \rho \xor x$.

Consequently, we have $\pnt_a$, $\pnt_b$, and $\pnt_c$ in $Q$.

\item[Pat8:] If $\val{v} = \set{\rho, a, \lambda_a, b}$,
By Proposition~\ref{appendix:val-path}~and~\ref{prop:one of path},
there is at most one node $x \in \set{a, b}$ such that $\pnt_x$ is not in $Q$.

For such $x$, we remove the left side occurrence of $v$ and all the movable occurrences of $\lambda_x$ and $\mu_x$.
We then add $\pnt_x \gets \rho \xor x$ to $Q$.

Consequently, we have $\pnt_a$ and $\pnt_b$ in $Q$.
\end{description}

After the above construction, we still have:
\[
\#_M(P') + \#_M(Q) \leq \#_M(P).
\]

\paragraph{Transferring the goal of $P$ without changing its size}\mbox{}

Now we define all the goals without the XOR cancellation in $Q$.

First, we consider the goals of the form $g_{a \lambda_a}$.
\begin{description}
\item[If $g_{a \lambda} \gets p \xor a \xor \lambda$:]
We move it from $P'$ to $Q$.

\item[If $g_{a \lambda} \gets t_1 \xor t_2$:]
Hereafter, for each subcase, we show that we have $\pnt_a$ in $Q$.
It suffices for our construction because
\begin{enumerate}
\item We delete the definition of $g_{a \lambda}$; it decrements the size of $P'$ by three.
\item We then add $g_{a \lambda} \gets \pnt_a \xor \lambda$ to $Q$; it increments the size of $Q$ by three.
\item Totally, $\#_M(P') + \#_M(Q) \leq \#_M(P)$ still holds.
\end{enumerate}

\item[Case $g_{a \lambda} \gets p \xor a$ ($\lambda$ has been removed to add $\pnt_a$ to $Q$):]
Clearly we have $\pnt_a$ in $Q$.

\item[Case $g_{a \lambda} \gets v \xor a$ and $\val{v} = \set{\lambda, \rho}$:]
By Pat1, we have $\pnt_a$ in $Q$.

\item[Case $g_{a \lambda} \gets \pnt_a \xor \lambda$:]
Clearly we have $\pnt_a$ in $Q$.

\item[Case $g_{a \lambda} \gets v \xor \rho$ and $\val{v} = \set{a, \lambda_a}$:]
By Proposition~\ref{appendix:val-path}~and~\ref{prop:one of path}, we have $\pnt_a$ in $Q$.

\item[Case $g_{a \lambda} \gets v \xor b$ and $\val{v} = \set{\rho, a, \lambda_a, b}$:]
By Pat8, we have $\pnt_a$ and $\pnt_b$ in $Q$.

\item[Case $g_{a \lambda_a} \gets v \xor g_{a \mu_a}$ and $\val{v} = \set{\lambda_a, \mu_a}$;]
By Pat2, we have $\pnt_a$ in $Q$.

\item[Case $g_{a \lambda_a} \gets v \xor g_{a b}$ and $\val{v} = \set{\lambda_a, b}$:]
By Pat3, we have $\pnt_a$ in $Q$.

\item[Case $g_{a \lambda_a} \gets v \xor g_{b \lambda_b}$ and $\val{v} = \set{a, \lambda_a, b, \lambda_b}$:]
By Pat6, we have $\pnt_a$ in $Q$;

\item[Case $g_{a \lambda_a} \gets v \xor g_{b c}$ and $\val{v} = \set{a, \lambda_a, b, c}$:]
By Pa7, we have $\pnt_a$ in $Q$;

\item[Case $g_{a \lambda_a} \gets v_1 \xor v_2$ and $v_1, v_2$ are variables:]
By Proposition~\ref{appendix:vv prop} and~\ref{prop:one of path}, we have $\pnt_a$ in $Q$.
\end{description}

Next, we consider the goals of the form $g_{a b}$ where $a$ and $b$ are (non-local) nodes.
\begin{description}
\item[If $g_{a b} \gets \rho \xor a \xor b$:] We copy it to $Q$.

\item[If $g_{a b} \gets t_1 \xor t_2$:]
Hereafter, for each subcase, we show that we have $\pnt_a$ in $Q$.
As the same as the construction of $g_{a \lambda_a}$, it suffices for out construction.

\item[If $g_{a b} \gets \pnt_a \xor b$:] Clearly, we have $\pnt_a$ in $Q$.

\item[If $g_{a b} \gets v \xor \rho$ and $\val{v} = \set{ a, b }$:]
By Proposition~\ref{appendix:val-path}~and~\ref{prop:one of path}, we have $\pnt_a$ or $\pnt_b$ in $Q$.

\item[If $g_{a b} \gets v \xor \lambda_a$ and $\val{v} = \set{ \rho, a, b, \lambda_a}$:]
By Proposition~\ref{appendix:val-path}~and~\ref{prop:one of path}, we have $\pnt_a$ or $\pnt_b$ in $Q$.

\item[If $g_{a b} \gets v \xor c$ and $\val{v} = \set{ \rho, a, b, c}$:]
By Proposition~\ref{appendix:val-path}~and~\ref{prop:one of path}, we have $\pnt_a$ or $\pnt_b$ in $Q$.

\item[If $g_{a b} \gets v \xor g_{a \lambda_a}$ and $\val{v} = \set{ b, \lambda_a}$:]
By Pat3, we have $\pnt_a$ in $Q$.

\item[If $g_{a b} \gets v \xor g_{a c}$ and $\val{v} = \set{ b, c }$:]
By Pat4, we have $\pnt_b$ and $\pnt_c$ in $Q$.

\item[If $g_{a b} \gets v \xor g_{c \lambda_c}$ and $\val{v} = \set{ a, b, c, \lambda_c }$:]
By Pat7, we have $\pnt_a$ or $\pnt_b$ in $Q$.

\item[If $g_{a b} \gets v \xor g_{c d}$ and $\val{v} = \set{ a, b, c, d }$:]
By Pat5, we have $\pnt_a$ or $\pnt_b$ in $Q$.

\item[If $g_{a b} \gets v_1 \xor v_2$:]
By Proposition~\ref{appendix:vv prop} and~\ref{prop:one of path}, we have $\pnt_a$ or $\pnt_b$ in $Q$.

\end{description}

By our construction, the following holds in $Q$.
\begin{lemma}\label{lemma:after normalization}\mbox{}
\begin{itemize}
\item $\result{Q} = \result{P}$.
\item $\#_M(Q) \leq \#_M(P)$.
\item Every temporal variable of $Q$ is the form of $\pnt_a \gets \rho \xor a$
where $a$ is a node of $G$.
\item The definition of each goal $g$ of $Q$ forms one of the following:
\begin{itemize}
\item $g \gets \pnt_a \xor x;$
\item $g \gets \rho \xor x \xor y;$
\end{itemize}
\end{itemize}
\end{lemma}

\subsubsection{Normalization IV: Finalization}\mbox{}\\
Finally, we prove Lemma~\ref{appendix main lemma}.

The above lemma Lemma~\ref{lemma:after normalization} does not ensure that
every edge $(a, b)$ is represented by $g_{a b} \gets \pnt_a \xor b$ or $g_{a b} \gets \pnt_b \xor a$
because $Q$ may not have both of $\pnt_a$ and $\pnt_b$.

In the case, $Q$ has the following definitions:
\[
\begin{array}{l}
g_{a b} \gets \rho \xor a \xor b; \\
g_{a \lambda} \gets \rho \xor a \xor \lambda; \\
g_{a \mu} \gets \rho \xor a \xor \mu;
\end{array}
\]
We change this part as follows without changing $\#_M(Q)$:
\[
\begin{array}{l}
\pnt_a \gets \rho \xor a ;\\
g_{a b} \gets \pnt_a \xor b; \\
g_{a \lambda} \gets \pnt_a \xor \lambda; \\
g_{a \mu} \gets \pnt_a \xor \mu;
\end{array}
\]

Repeatedly applying this modification, we transform $Q$ to $Q'$ that satisfies all the conditions of Lemma~\ref{appendix main lemma}.

\newpage

\end{document}